\documentclass[a4paper,UKenglish,cleveref, autoref, thm-restate]{lipics-v2021}

\title{Novel Complexity Results for Temporal Separators with Deadlines} 

\author{Riccardo Dondi}{Università degli Studi di Bergamo, Italy}{riccardo.dondi@unibg.it}{https://orcid.org/0000-0002-6124-2965}{}

\author{Manuel Lafond}{Université de Sherbrooke, Canada}{manuel.lafond@usherbrooke.ca}{https://orcid.org/0000-0002-5305-7372}{}

\authorrunning{R. Dondi, M. Lafond} 

\Copyright{Riccardo Dondi, Manuel Lafond} 



\category{} 

\relatedversion{} 

\nolinenumbers 

\EventEditors{John Q. Open and Joan R. Access}
\EventNoEds{2}
\EventLongTitle{42nd Conference on Very Important Topics (CVIT 2016)}
\EventShortTitle{CVIT 2016}
\EventAcronym{CVIT}
\EventYear{2016}
\EventDate{December 24--27, 2016}
\EventLocation{Little Whinging, United Kingdom}
\EventLogo{}
\SeriesVolume{42}
\ArticleNo{23}

\usepackage{float}
\usepackage{tikz}
\usepackage{verbatim}
\usetikzlibrary{decorations.pathreplacing,positioning, arrows.meta,shapes}
\usetikzlibrary{patterns}

\tikzset{main node/.style={circle,fill=blue!20,draw,minimum size=0.7cm,inner sep=0pt},
            }
\tikzset{example node/.style={circle,fill=black,draw,minimum size=0.2cm,inner sep=0pt},
}

\usepackage{todonotes}
\usepackage[linesnumbered,boxed]{algorithm2e}
\SetKwRepeat{Do}{do}{until}

\newcommand{\DMC}{\textsc{Directed Multicut}}

\newcommand{\TempSepN}{\textsc{(s,z)-Temporal Separator}}

\newcommand{\TempSep}[1]{\textsc{(s,z,#1)-Temporal Separator}}

\newcommand{\StTempSep}[1]{\textsc{(s,z,#1)-Strict Temporal Separator}}

\newcommand{\TempCutN}{\textsc{(s,z)-Temporal Cut}}

\newcommand{\TempCut}[1]{\textsc{(s,z,#1)-Temporal Cut}}

\newcommand{\UniformHypergraph}{\textsc{Vertex Cover $k$-Uniform Hypergraph}}

\newtheorem{problem}{Problem}

\newcommand{\uout}[1]{{#1}^-}
\newcommand{\uin}[1]{{#1}^+}

\begin{document}

\maketitle

\begin{abstract}
We consider two variants, \TempSep{$\ell$}
and \TempCut{$\ell$}, respectively, of the vertex separator and the edge cut problem in temporal graphs. The
goal is to remove the minimum number of vertices (temporal edges, respectively) in order to delete all the temporal paths that have time travel at most $\ell$
between a source vertex $s$ and target vertex $z$. First, we solve an open problem in the literature
showing that \TempSep{$\ell$} is NP-complete even when
the underlying graph has pathwidth bounded by four.
We complement this result showing that \TempSep{$\ell$}
can be solved in polynomial time for graphs of pathwidth
bounded by three. Then we consider the approximability of
\TempSep{$\ell$} and we show that it cannot be approximated within factor 
$2^{\Omega(\log^{1-\varepsilon}|V|)}$ for any constant $\varepsilon> 0$, unless $NP \subseteq ZPP$ ($V$ is the
vertex set of the input temporal graph) and that
the strict version is approximable within factor $\ell-1$
(we show also that it is unliklely that this factor can be improved).
Then we consider the \TempCut{$\ell$} problem, we show
that it is APX-hard and we present a $2 \log_2(2\ell)$ approximation algorithm.
\end{abstract}

\section{Introduction}
\label{sec:intro}

A central problem in checking network robustness
is finding the minimum number of vertices or edges
that needs to be removed in order to disconnect the
network. In classic (static) graphs this is modeled 
by computing a minimum cut or a minimum vertex separator
between a source vertex $s$ and a target vertex $z$.
The static graph model however does not consider how
a network may change over time. In transportation
networks, for example, the time schedule is a fundamental aspect that has to be taken into account for analyzing several properties, like connectedness and robustness.
The need to incorporate time information of edge availability has led to the introduction of the
\emph{temporal graph} model~\cite{DBLP:journals/jcss/KempeKK02,holme2015modern,DBLP:journals/im/Michail16}, where edges are assigned
timestamps in a discrete set that define when each edge is available
(thus when a specific transport is available in a transportation network).

The robustness problems (minimum cut and
minimum vertex separator) have been considered also
for temporal graphs, where the paths to be removed
have to be \emph{temporal}, that is they have to satisfy 
a time constraint.
Given a source $s$ and a target $z$,
the \TempCutN{} problem asks for the minimum number of
temporal edges that have to be removed so that
$s$ and $z$ are disconnected, while 
\TempSepN{} asks for the minimum number of vertices that have to be removed so that
$s$ and $z$ are disconnected.
\TempCutN{} is known to be in solvable in polynomial time~\cite{DBLP:journals/networks/Berman96},
\TempSepN{} is known NP-hard~\cite{DBLP:journals/jcss/KempeKK02,DBLP:journals/jcss/ZschocheFMN20},
and its fixed-parameter tractability and approximability (and variants thereof) have been studied~\cite{DBLP:journals/tcs/FluschnikMNRZ20,DBLP:journals/jcss/MaackMNR23,DBLP:conf/isaac/HarutyunyanKP23,DBLP:journals/aamas/KlobasMMNZ23,DBLP:journals/jcss/IbiapinaS24}.

A variant of the \TempSepN{} problem that has been considered in~\cite{DBLP:conf/isaac/HarutyunyanKP23} to model
the robustness of transportation system, defines $s$
and $z$ separated if
the time needed to move from $s$ to $z$ is above
a time threshold. The motivation is that if the time
to move from $s$ to $z$ is increased considerably,
then this makes the possibility of moving from $s$ 
to $z$ unlikely.
This is modeled in~\cite{DBLP:conf/isaac/HarutyunyanKP23} 
by defining the \TempSep{$\ell$} problem, that asks for a smallest subset of vertices whose removal deletes each temporal path 
that takes at most time $\ell$ between $s$ and $z$.
Several results have been given for \TempSep{$\ell$} in~\cite{DBLP:conf/isaac/HarutyunyanKP23}.
The \TempSep{$\ell$} problem is NP-complete when $\ell=1$
and the temporal graph is defined over (at least) two timestamps.
On the other hand, the problem is solvable
in polynomial time when the underlying graph is a tree (after deleting $s$ and $z$)  and when it has branchwidth at most two. 
As for the approximation complexity, it is shown
to be not approximable within factor $\Omega(\ln |V|+ \ln \tau)$ assuming that $NP \not\subset DTIME(|V|^{\log \log |V|}$) ($V$ is the set of vertices of the temporal graph, $\tau$ the number of timestamps); moreover, 
a $\tau^2$-approximation algorithm is presented 
(also a $\tau$-approximation algorithm for \TempSepN{}).
Finally, it is shown that solving \TempSep{$\ell$} 
when the underlying graph has bounded pathwidth is at least as difficult as solving a problem called 
\textsc{Discrete Segment Covering}, whose complexity is unsolved~
\cite{DBLP:journals/siamdm/BergrenEGK22}.  
In Section~\ref{sec:BuondedWidth}, we solve the status of \TempSep{$\ell$} on graphs of bounded pathwidth: we show that the problem is NP-complete even with $\ell=1$ and when the underlying graph has pathwidth at most $4$, and we give a polynomial-time algorithm when the pathwidth is at most $3$.
Then  in Section~\ref{sec:approx}, 
we show that \TempSep{$\ell$} cannot be approximated within factor 
$2^{\Omega(\log^{1-\varepsilon}|V|)}$ for any constant $\varepsilon> 0$, unless $NP \subseteq ZPP$
and we present an $\ell-1$-approximation algorithm for the strict variant of \TempSep{$\ell$}. We show also that improving this factor is a challenging problem, since
the strict variant of \TempSep{$\ell$} is hard to approximate as 
\UniformHypergraph{}, where $k = \ell-1$.
In Section~\ref{sec:cut}, we consider the \TempCut{$\ell$} problem and we show
that it is APX-hard, 
which contrasts with the polynomial-time solvability of the problem when deadlines are not considered, and we present a $2 \log_2(2\ell)$-approximation algorithm.
In Section~\ref{sec:prel} we give some definitions and we formally define the two problems we are interested into.
Some of the proofs are omitted due to space constraint.  

\section{Preliminaries}
\label{sec:prel}

For an integer $n$, we use the notation $[n] = \{1,2,\ldots,n\}$.  A \emph{temporal graph} $G=(V,E,\tau)$ is defined over a set $V$ of vertices and  a set $E \subseteq V \times V \times [\tau]$ of temporal edges, where $\tau \in \mathbb{N}$. 
An \emph{undirected} edge in a temporal graph is then a triple 
$(u,v,t)$, where $u,v \in V$  and
$t \in [\tau]$ is a timestamp\footnote{As in \cite{DBLP:conf/isaac/HarutyunyanKP23} we assume that 
$(u, v, t)$ is a temporal edge of $G$ if and only if 
$(v, u, t)$ is a temporal edge of $G$.}.
Note that we denote by $uv$ an edge in an undirected
(static) graph and $(u,v)$ an arc in a directed (static)
graph.

We say that $u, v$ are \emph{neighbors} in $G$ if they share some temporal edge, and denote the set of neighbors of $u$ by $N_G(u)$ (we drop the subscript if clear).
Given a set $V' \subseteq V$,
we denote by $G[V']$ the subgraph induced by $V'$, which contains vertex set $V'$ and every temporal edge whose two endpoints are in $V'$. 
We also denote by $G - V' = G[V\setminus V']$ the temporal graph
obtained by removing each vertex in $V'$.
Given a set $E' \subseteq E$,
we denote by $G - E'$ 
the temporal graph
obtained by removing each temporal edge in $E'$.

An interval $[t_1, t_2]$, 
with $t_1, t_2 \in [\tau]$ and $t_1 \leq t_2$, is the sequence of consecutive timestamps between $t_1$ and $t_2$.
Given interval $[t_1, t_2]$, 
$G([t_1, t_2])$
is the temporal subgraph of $G$ that has temporal edges having timestamps
in $[t_1,t_2]$.

A \emph{temporal path} $P$ in a temporal graph $G$
is a sequence of temporal edges such that
\[P = [(u_1, v_1, t_1), (u_2, v_2, t_2), \dots, 
(u_h, v_h, t_h)],
\]
with $v_i = u_{i+1}$
and $t_i \leq t_{i+1}$, for each $i \in [h-1]$,
and $u_i \neq u_j$, $v_i \neq v_j$, for each $i,j \in [h]$ with $i \neq j$.
If $t_i < t_{i+1}$, with $i \in [h-1]$, then the temporal
path is \emph{strict}.
Given a temporal path $P = [(u_1, v_1, t_1), (u_2, v_2, t_2), \dots, 
(u_h, v_h, t_h)]$, the \emph{travelling time} $tt(P)$ of $P$
is defined as $tt(P) = t_h - t_1+1$.
The set of vertices of a temporal path $P$ is denoted
by $V(P)$.  
We may refer to a temporal path of travelling time \emph{at most} $\ell$ as an $\ell$-temporal path.
Two temporal paths are \emph{temporal edge disjoint}
if they don't share any temporal edge.

Given a temporal graph $G=(V,E, \tau)$,
we assume that there are two special 
vertices $s,z \in V$, which are
the source and the target vertex of $G$.
A set $V' \subseteq V$ is an $(s,z)$-temporal separator ($(s,z)$-strict temporal separator, respectively) in $G$ if there is no temporal path (strict
temporal path, respectively) in $G[V \setminus V']$ 
between $s$ and $z$.
Given $\ell \in [\tau]$, a set $V' \subseteq V$ is an $(s,z,\ell)$-temporal separator ($(s,z,\ell)$-strict temporal separator, respectively) in $G$ if
there is no temporal path (strict
temporal path, respectively) between $s$ and $z$ 
in $G[V \setminus V']$ of travelling 
time at most $\ell$.

A set $E' \subseteq E$ of temporal edges is an $(s,z)$-temporal cut ($(s,z)$-strict temporal cut, respectively) in $G$ if there is no temporal path (strict
temporal path, respectively) in $G - E'$  between $s$ and $z$.
Given $\ell \in [\tau]$ a set $E' \subseteq E$ is an $(s,z,\ell)$-temporal cut ($(s,z,\ell)$-strict temporal cut, respectively) in $G$ if there is no temporal path 
(strict temporal path, respectively) in $G- E'$ between $s$ and $z$ of travelling time at most $\ell$.

Now, we are ready to define
the combinatorial problems we are interested into.
The first, introduced in~\cite{DBLP:conf/isaac/HarutyunyanKP23}, is the following.

\begin{problem} (\TempSep{$\ell$}) \\ 
\textbf{Input:} a 
temporal graph $G=(V,E,\tau)$, two vertices $s,z \in V$, a positive integer $\ell \in [\tau]$.\\
\textbf{Output:} An (s,z,$\ell$)-temporal separator $V' \subseteq V$ in $G$ of minimum size.
\end{problem}

We denote by \StTempSep{$\ell$} the variant of 
\TempSep{$\ell$} where we look for an (s,z,$\ell$)-strict temporal separator.

We consider now the second problem we are interested into.

\begin{problem} (\TempCut{$\ell$}) \\ 
\textbf{Input:} a 
temporal graph $G=(V,E,\tau)$, two vertices $s,z \in V$, a positive integer $\ell \in [\tau]$.\\
\textbf{Output:} An (s,z,$\ell$)-temporal cut $E' \subseteq E$ in $G$ of minimum size.
\end{problem}

Note that if $\ell = \tau$, then
\TempSep{$\ell$} (\TempCut{$\ell$}, respectively)
is exactly the \TempSepN{} problem (the \TempCutN{}) problem, respectively) where we look
for an $(s,z)$- temporal separator ($(s,z)$-temporal cut, respectively) in $G$.






\section{Graphs of bounded pathwidth}
\label{sec:BuondedWidth}

Let us first recall the notion of pathwidth.
A \emph{nice path decomposition} of a graph $G$ is a sequence of set of vertices $(X_1, \ldots, X_q)$, with each $X_i \subseteq V(G)$ referred to as a \emph{bag}, such that all of the following holds:
\begin{itemize}
    \item 
    $X_1 = X_q = \emptyset$;

    \item 
    for $i \in \{2, \ldots, q\}$, either $X_i = X_{i-1} \cup \{v\}$, in which case $X_i$ is called an \emph{introduce bag}; or $X_i = X_{i-1} \setminus \{v\}$, in which case $X_i$ is a \emph{forget bag}.

    \item 
    for any pair of vertices $u$ and $v$ that share an edge in $G$, some bag $X_i$ contains both $u$ and $v$ (regardless of the time of the edge).
    
    \item 
    for any vertex $v \in V(G)$, there are $i, j \in [q - 1]$ such that the set of bags that contain $v$ is precisely $X_i, X_{i+1}, \ldots, X_j$.
\end{itemize}
The \emph{width} of the nice path decomposition is $\max_{i \in [q-1]} (|X_i| - 1)$.  The \emph{pathwidth} of $G$ is the minimum width of a nice path decomposition of $G$.

We show that \TempSep{$\ell$} is NP-complete even on graphs of pathwidth $4$ and when $\ell = 1$.
Our reduction is from the \textsc{Multicolored Independent Set} problem, where we are given a (static) graph $G$ and a partition $\{V_1, \ldots, V_k\}$ of $V(G)$ into $k$ sets.  The $V_i$ sets are called \emph{color classes}.  The question is whether there is an independent set $I$ of $G$ such that $|I \cap V_i| = 1$ for each $i \in [k]$, that is, we must choose one vertex per color class to form an independent set.  
Note that this problem is typically used to prove W[1]-hardness, but it is also NP-hard~\cite{fellows2009parameterized}\footnote{Note 
that the problem considered in~\cite{fellows2009parameterized}, 
\textsc{Multicolored Clique}, is equivalent
to \textsc{Multicolored Independent Set} if we consider complementary relations (edges and no edges) between
vertices of different color classes of the input graph.}.

Let $G$ be an instance of \textsc{Multicolored Independent Set}, with vertex partition $V_1, \ldots, V_k$.  We assume, without loss of generality, that $|V_i| = n$ for every $i \in [k]$.

Let us construct from $G$ an instance of the \TempSep{$\ell$} problem consisting of temporal graph $H$, vertices $s$ and $z$ to separate, and $\ell = 1$.  
In the construction, we  assign to temporal edges of $H$ a time in the set $\tau = \{t_{green}, t_{red}\} \cup \{t_u : u \in V(G)\}$, with the understanding that all elements of $\tau$ correspond to a distinct integer.  
Since $\ell = 1$, one can view the temporal edges as being colored by an element of $\tau$ and the problem as having to destroy all monochromatic paths --- so the subscripts of the elements of $\tau$ can be seen as colors.  
So from $G$, first add to $V(H)$ the vertices $s, z$, and a new vertex $r$. 
We then have vertex-choosing gadgets and edge-verification gadgets (see Figure~\ref{fig:allstuff} for an illustration).

\begin{figure}[b]
    \centering
    \includegraphics[width=0.95\linewidth]{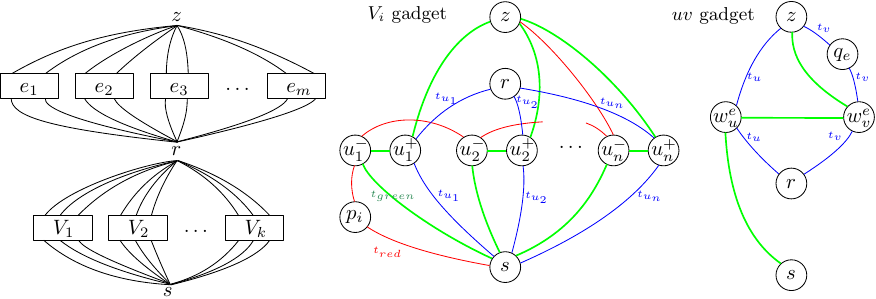}
    \caption{Top left: main structure of the reduction: there is a vertex-chossing phase leading to paths to $r$, followed by an edge vertification phase (note, not all edges are shown).  Middle: a $V_i$ gadget, with $t_{green}$ edges in green, $t_{red}$ edges in red, and $t_u$ edges in blue.  Right: an edge gadget for $e = uv$, with $t_{green}$ edges in green and $t_u, t_v$ edges in blue.}
    \label{fig:allstuff}
\end{figure}

\medskip

\noindent 
\textbf{Vertex-choosing gadgets.}
For each $i \in [k]$, build a gadget for $V_i$ as follows.
For every $u \in V_i$, add to $H$ two vertices $\uout{u}$ and $\uin{u}$. 
Then add a temporal path at time $t_{green}$ from $s$ to $z$ formed by the temporal edges $(s, \uout{u}, t_{green}), (\uout{u}, \uin{u}, t_{green}), (\uin{u}, z, t_{green})$.
 Also add a temporal path at time $t_u$ from $s$ to $r$ formed by the temporal edges $(s, \uin{u}, t_u), (\uin{u}, r, t_u)$.

To complete the gadget, add a new vertex $p_i$, then add temporal edges so that there is a path with a time of $t_{red}$ from $s$ to $z$ that first goes through $p_i$, then through all the $\uout{u}$ vertices exactly once. 
More precisely, denote $V_i = \{u_1, \ldots, u_n\}$, where the ordering is arbitrary.  Then, add the temporal edges $(s, p_i, t_{red}), (p_i, \uout{u}_1, t_{red})$ and $(\uout{u}_n, z, t_{red})$, and for each $h \in [n - 1]$, add the temporal edge 
$(\uout{u}_h, \uout{u}_{h+1}, t_{red})$.  


\medskip

\noindent 
\textbf{Edge verification gadgets.}
Next, we construct gadgets to verify that chosen vertices correspond to an independent set of $G$.
For each edge $e = uv$ of $G$, where $u \in V_i$ and $v \in V_j$ such that $i < j$, add to $H$ two vertices 
$w_u^e$ and $w_v^e$.  
Add a temporal path at time $t_{green}$ from $s$ to $z$ formed by the temporal edges $(s, w_u^e, t_{green}), (w_u^e, w_v^e, t_{green})$, and $(w_v^e, z, t_{green})$, enforcing the deletion of at least one of the two vertices.
Next, add a temporal path at time $t_u$ formed by the temporal edges $(r, w_u^e, t_u), (w_u^e, z, t_u)$.  
Finally, add a new vertex $q_e$ and a temporal path at time $t_v$ going from $r$ to $z$ formed by the temporal edges $(r, w_v^e, t_v), (w_v^e, q_e, t_v), (q_e, z, t_v)$.




   
\begin{restatable}{lemma}{lemboundedpathwidth}
\label{lemboundedpathwidth}
    The graph $H$ constructed from $G$ as described above  has pathwidth at most $4$.
\end{restatable}

The idea is that we can get a path decomposition by putting $s, z, r$ in every bag.  Then, the $V_i$ and the $uv$ gadgets are easy to construct using two extra vertices per bag.  We next show that NP-hardness holds.

\begin{restatable}{lemma}{lempwhard}
\label{lem:pwhard}
    The graph $G$ has a multicolored independent set if and only if $H$ has an $(s, z, 1)$-temporal separator of size $|V(G)| + |E(G)|$.
\end{restatable}
\noindent 
\emph{Proof sketch.}
    Let $I = \{u_1, \ldots, u_k\}$ be a multicolored independent set of $G$, with each $u_i \in V_i$.  In $H$, in the $V_i$ gadget, delete every $u^+$ vertex, \emph{except} $u_i^+$, and delete $u_i^-$ instead.  This removes all the $t_{green}$ and the $t_{red}$ temporal paths, but $s$ can reach $r$ with a temporal path at time $t_{u_i}$. 
 Then for each edge $e = uv$ of $G$, in the $uv$ gadget, delete $w_u^e$ if $s$ reaches $r$ with time $t_u$, and delete $w_v^e$ otherwise.  This removes the $t_{green}$ temporal path, and since we delete at least one of $u^+$ or $v^+$, there remains no temporal path at time $t_u$ or $t_v$ going through the gadget.  

 Conversely, suppose there is an $(s, z, 1)$-temporal separator in $H$ of size $|V(G)| + |E(G)|$.  The $t_{green}$ temporal paths enforce deleting, for each $u \in V(G)$, one of $u^-$ or $u^+$, and also for each $e = uv \in E(G)$ one of $w_u^e$ or $w_v^e$.  There is no room for other deletions.  Because of the $t_{red}$ temporal path in the $V_i$ gadget, some $u^-$ is deleted and some $u^+$ is kept.  Also, we cannot keep $u^+$ and $v^+$ from different $V_i, V_j$ gadgets if $uv \in E(G)$, as otherwise there will be a $t_u$ or $t_v$ temporal path going through the $uv$ gadget.  Hence, the kept $u^+$ vertices correspond to a multicolored independent set.
\begin{theorem}
    The \TempSep{$\ell$} problem is NP-complete, even with $\ell = 1$ and on temporal graphs of pathwidth at most $4$, 
    and even if each edge is present in only one timestamp.
\end{theorem}
We mention that it should be possible to modify the proof to prove the same hardness for the \emph{strict} variant of the problem, as it suffices to replace each label $t_{green}, t_{red}, t_u$ with time intervals that are far enough from each other.  We leave the details for a future version.


\subsection*{Graphs of pathwidth $3$}

Here we show that \TempSep{$\ell$} can be solved in polynomial time on graphs of pathwidth $3$, for any $\ell$, which shows that the above hardness is tight.  Note that now, we allow temporal edges to have multiple activation times, and we allow them to be directed or not.  
Let $G$ be a temporal graph.  We first apply two reduction rules, which can easily seen to be safe.

\medskip 

\noindent 
\textbf{Rule 1.}  If $G$ has a vertex $v $ such that 
$[(s,v,t) (v,z,t')$ is an $\ell$-temporal path, then delete $v$.

\medskip 

\noindent 
\textbf{Rule 2.} If $G - \{s, z\}$ has multiple connected components $C_1, C_2, \ldots, C_p$, then solve each subgraph $G[C_1 \cup \{s, z\}], \ldots, G[C_p \cup \{s, z\}]$ separately.

\medskip 

We say that $G$ is \emph{clean} if none of the above rules is applicable to $G$.
Note that for $n = |V(G)|$ and $m = |E(G)|$, one can determine in time $O(n + m)$ whether one of the rules applies, and each rule can be applied at most $n$ times, and so a graph can be made clean in time $O(n^2 + nm)$.

We now turn to (clean) graphs of pathwidth at most $3$.  The idea is to first check whether there is an $(s, z)$-separator of size at most $3$  (so, a separator that ignores edge times).  If there is one, then there is an $(s, z, \ell)$-temporal separator of size at most $3$.  In this case we can compute 
a solution by trying every combination of at most three vertices.  If there is no such separator, we can rely on the structural lemma below, which is illustrated in Figure~\ref{fig:pw3}.
First, we can show that there is a sequence
of bags that contain both $s$ and $z$; this allows
to find $B$, consisting 
of the vertices introduced and forgotten within this sequence.  This $B$ induces a caterpillar\footnote{Recall that a \emph{caterpillar} is a tree in which there is a \emph{main path} $w_1 - \ldots - w_k$, and every vertex not on that path is a leaf adjacent to a vertex of the path.
}, as removing $s$ and $z$ from these bags yields a subgraph of pathwidth at most $2$.
The vertices other than $s$ or $z$ in bags that precede this sequence form $A$, and the vertex $s$ or $z$ introduced last has at most one neighbor in $A$ (e.g., in Figure~\ref{fig:pw3}, $z$ has only $y$ as a neighbor in $A$).  Analogously, $C$ contains vertices of bags that occur after the sequence.

\begin{figure}[h]
    \centering
    \includegraphics[width=0.6\linewidth]{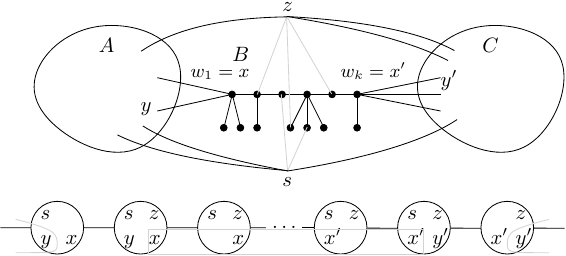}
    \caption{An illustration of the structure described by Lemma~\ref{lem:pwstructure}. 
 The top part shows the temporal graph, the bottom part shows the bags that contain both $s$ and $z$ (the left and right areas represent $A$ and $C$, and the middle box represents $B$).}
    \label{fig:pw3}
\end{figure}

\begin{restatable}{lemma}{pwstrcuture}
\label{lem:pwstructure}
    Suppose that $G$ is clean, has pathwidth $3$, and has no $(s, z)$-separator of size $3$ or less.
    Then $V(G) \setminus \{s, z\}$ can be partitioned into three non-empty sets $A, B, C$ such that all of the following holds: 
	\begin{itemize}
		\item 
		either $s$ has at most one neighbor in $A$ (resp. $C$), or $z$ has at most one neighbor in $A$ (resp. $C$).
		
		\item 
		the subgraph induced by $B$ is a caterpillar with main path $w_1 - w_2 - \ldots - w_k$.  Moreover, $w_1$ is the only vertex of $B \cup C$ whose neighborhood can intersect with $A$, and $w_k$ is the only vertex of $B \cup A$ whose neighborhood can intersect with $C$.
		 
	\end{itemize}
\end{restatable}

One can infer from this structure that ``most'' deletions occur on the caterpillar path.

\begin{restatable}{lemma}{pwdels}
\label{lem:pwdels}
    Suppose that $G$ is clean and that $V(G) \setminus \{s, z\}$ can be partitioned into the sets $A, B, C$ as in Lemma~\ref{lem:pwstructure}.
    Then there exists a minimum $(s, z, \ell)$-temporal separator in which at most one vertex of $A$ is deleted, at most one vertex of $C$ is deleted, and all other deleted vertices are on the main path $w_1 - \ldots - w_k$ of the $B$ caterpillar.
\end{restatable}
\begin{algorithm}[H]
\DontPrintSemicolon
\SetKwProg{Fn}{function}{}{}
\Fn{getTemporalSeparator($G, s, z, (A, B, C)$)}
  {
    //We assume that $G$ is clean\;
    Let $\mathcal{D} = \{ D \subseteq V(G) : \mbox{$D \cap A \leq 1, D \cap C \leq 1$, and $D \cap B = \emptyset$}\}$\;
    
    \For{each $D \in \mathcal{D}$}{
        $D' = extendSeparator(G, s, z, (A, B, C), D$)   \tcp*{pass copy of $D$}
    }
    return a smallest $(s, z, \ell)$-separator $D'$ found\;
  }
\;
\Fn{extendSeparator($G, s, z, (A, B, C), D$)}
{
    \uIf{$G - (D \cup B)$ has an $(s, z, \ell)$-temporal path}{
        return ``impossible''
    }
    Let $w_1 - \ldots - w_k$ be the main path of the caterpillar $B$\;
    \For{$i = 1, 2, \ldots, k - 1$}
    {
        Define $W_i = A \cup \{s, z\} \cup (N[w_1, \ldots, w_i] \setminus \{w_{i+1}\})$\;
        \uIf{$G[W_i] - D$ contains an $(s, z, \ell)$-temporal path}
        {
            Add $w_i$ to $D$\;  
        }
    }
    \uIf{$G - D$ contains an $(s, z, \ell)$-temporal path}{
        Add $w_k$ to $D$\;
    }
    return $D$;
}
  \caption{Main algorithm}
  \label{alg:pwalgo}
\end{algorithm}

Our algorithm proceeds as follows.  We first make $G$ clean, possibly handling multiple connected components of $G - \{s, z\}$ separately, and construct the sets $A, B, C$ from Lemma~\ref{lem:pwstructure}.
We then run Algorithm~\ref{alg:pwalgo}, which first guesses every way that a solution could delete at most one vertex from $A$ and at most one vertex from $C$, noting that such a solution exists by Lemma~\ref{lem:pwdels}.
There are $O(n^2)$ possible guesses, and for each of them, we solve a restricted version of the problem where we can only delete vertices from the caterpillar $B$.  

This is done by the $extendSeparator$ function, which attempts to extend the guess by finding the minimum number of deletions to do on the main $B$ path $w_1 - \ldots - w_k$, traversing it from left to right.
When at a specific $w_i$, we restrict the graph to $G[W_i] - D$, where $D$ contains the deletions made so far and $W_i$ contains $A$, $w_i$ and the predecessors of $w_i$, and their neighbors.  In a greedy manner, we delete $w_i$ only if it creates a temporal path in this restricted graph, considering the deletions $D$ applied so far. 
This ensures that there is no temporal path that goes through the $w_i$'s in increasing order (possibly going in $A$ as well), and we make one last check at the end to ensure that no temporal path goes through $C$ and then on the path in the reverse order (if so, we delete $w_k$ to prevent that).
\begin{restatable}{theorem}
{thmpwthree}
\label{thm:pw3}
    The \TempSep{$\ell$} problem can be solved in time $O(n^4 m)$ on graphs on pathwidth at most $3$, where $n = |V(G)|$ and $m = |E(G)|$.
\end{restatable}

Note that the above algorithm should work for the strict variant of the problem. 
 Indeed, the algorithm only needs to query whether a forbidden path exists, so it suffices to have access to a routine that determines whether a strict $\ell$-temporal path exists.  Again we leave the details for a future version.


\section{Approximation of \TempSep{$\ell$}} 
\label{sec:approx}

\subsection{Hardness of Approximation}
\label{subsec:HardApprox}

In this section we strengthen the inapproximability
of \StTempSep{$\ell$} (and later we extend the inapproximability to \TempSep{$\ell$}). 
We prove the result by giving an approximation preserving
reduction from \DMC{} to \StTempSep{$\ell$}.
\DMC{} is known to be inapproximable within
factor $2^{\Omega(\log^{1-\varepsilon}|N|)}$, for
any constant $\varepsilon > 0$,
even for a directed 
acyclic graph with a set $N$ of vertices, unless $NP \subseteq ZPP$~\cite{DBLP:journals/jacm/ChuzhoyK09}.
We recall here that \DMC{} (we consider it is defined on a directed 
acyclic graph),
%
given a directed acyclic graph $D=(N,A)$, a set $R$ of pairs $(s_1, z_1), \dots, (s_h,z_h)$ of vertices, with $s_i, z_i \in N$,
$i \in [h]$, 
asks for a minimum cardinality subset $A' \subseteq A$ so that each pair $(s_i,z_i)$, $i \in [h]$, is disconnected in $D- A'$.
Note that we assume that, for each
$(s_i, z_i) \in R$, $i \in [h]$, it holds that 
$s_i \neq t_i$, otherwise the vertex needs to 
be removed in order to separate the pair.
Given an instance $(D,R)$ of \DMC{}, 
the vertices that belong to a pair in $R$
are called terminals; 
the set $R_S$ ($R_Z$, respectively) contains those terminals
$v$ such that $v=s_i$ and 
$(s_i, z_j) \in R$
($v = z_j$ and $(s_i, z_j=v) \in R$, respectively).

Consider an instance $(D,R)$ of \DMC{}, in the following we construct
a corresponding instance 
of \StTempSep{$\ell$}.
We first present the idea of the construction.
Essentially $G$ contains $|A|+1$ copies
of each vertex, plus one vertex for each
arc in $A$, in addition to $s$ and $z$.
This ensures that only the vertices
associated with arcs will be removed,
as removing $|A|+1$ copies of each vertex
of $D$ requires to delete too many vertices.
As for the temporal edges of $G$, 
for each arc from $u$ to $v$ in $D$,
there is a path of length two, consisting of
a temporal edge from
the vertex associated with $u$ to
vertex associated with arc $(u,v)$
and 
a temporal edge from the
vertex associated with arc $(u,v)$
to the vertex associated with $v$.
Each vertex $v_i$ is associated
with a interval of $\ell$ timestamps 
$[\ell\cdot (2i-2), \ell \cdot (2i-1)-1]$.
The timestamps are assigned to temporal
edges of $G$ so that if there exists a path
in $D$ from $v_i$ to $v_j$, then there
exists a temporal path from two corresponding 
vertices of $G$, that is $x_{i,q}$ to $x_{j,r}$, 
$q,r \in [|A|+1]$, with timestamps in 
the interval $[\ell\cdot (2i-2), \ell \cdot (2i-1)-1]$.

Now, we present the formal definition of
the instance of \StTempSep{$\ell$}.
First, $\ell = 2|N|$.
Given $N = \{ v_1, \dots, v_{|N|} \}$,
since $D$ is a DAG, we assume that the vertices
of $D$ are sorted according to a topological sorting
of $D$, that is for $(v_i, v_j ) \in A$, $i,j \in [|N|]$, it holds that $i < j$. 
We define the set $V$ of vertices: 
$$V  = 
  \{ x_{i,q}: v_i \in N, q \in [|A|+1]\}\ \cup 
 \{ y_{i,j} : (v_i,v_j) \in A, i < j \} \ \cup 
 \{s,z\}.$$
Now, we define the set of temporal edges.
For a vertex $v_i \in N$, $i \in [|N|]$,  $Reach(D,v_i)$ 
denotes the set of vertices that can be reached with 
a path that starts from $v_i$ in $D$.
\begin{equation*} 
\label{eq2}
\begin{split}
E  = & \{ 
(s,x_{i,q},\ell \cdot (2i-2)) : v_i \in R_S, q \in [|A|+1] 
\}\ \cup \\
& \{
(x_{j,q},z,\ell \cdot (2i -1)-1) : v_j  \in R_Z,
v_j \in Reach(D,v_i), 
\exists (v_i, v_j) \in R, 
q \in [|A|+1] 
\} \ \cup \\
 & \{
(x_{i,q},y_{i,j},\ell \cdot (2h-2) + 2i-1): v_i \in N, v_i \in Reach(D, v_h),
  (v_i,v_j) \in A, h \in [i], q \in [|A|+1]\}\ \cup \\
 & \{
(y_{i,j},x_{j,q},\ell \cdot (2h-2) + 2(j-1)): 
 v_j \in N, v_j \in Reach(D,v_h),
 (v_i,v_j) \in A, h \in [i], q \in [|A|+1]\}\\
\end{split}
\end{equation*}
Now, we prove the relations between \DMC{} 
and \StTempSep{$\ell$}.

\begin{restatable}{lemma}{lemDMCST}
\label{lem:DMCSt} 
Consider an instance $(D,R)$ of \DMC{} and
a corresponding instance $(G,s,t, \ell)$ of \StTempSep{$\ell$}.
Then (1) given a solution $A'$ of \DMC{} on instance $(D,R)$ consisting of $k$ arcs we can compute
in polynomial time a solution of \StTempSep{$\ell$} 
on instance $(G,s,t, \ell)$
consisting of $k$ vertices; 
(2) given a solution of \StTempSep{$\ell$} 
on instance $(G,s,z, \ell)$ consisting of $k$ vertices
we can compute in polynomial time 
a solution $A'$ of \DMC{} on instance $(D,R)$ consisting of $k$ arcs.
\end{restatable}
\noindent 
\emph{Proof sketch.} For (1), given a solution $A'$ of \DMC{} on instance $(D,R)$ consisting of $k$ arcs, we can compute
in polynomial time a solution of 
\StTempSep{$\ell$} on the corresponding instance 
$(G,s,t, \ell)$ as follows:
$V' = \{ y_{i,j}: (v_i,v_j) \in A'\}$.
Indeed, if
$G[V \setminus V']$ contains a strict temporal path $p$ it must
pass through two vertices $x_{i,r}, x_{j,r}$ and there 
is a path between the corresponding vertices 
$v_i$, $v_j$ of $D-  A'$
such that
$(v_i, v_j) \in R$.
For (2), given a solution $V'$ of \StTempSep{$\ell$} 
on instance $(G,s,t,\ell)$, where $|V'|= k$, first we
can show that $V'$ contains only vertices $y_{i,j}$.
We can define a solution $A'$ of \DMC{} on instance $(D,R)$ 
as follows $A' = \{ (v_i,v_j): y_{i,j}  \in V'\}$.


It follows from Lemma \ref{lem:DMCSt}
that we have designed an approximation
preserving reduction from \DMC{} to
\StTempSep{$\ell$}.
Since \DMC{} is not approximable within factor
$2^{\Omega(\log^{1-\varepsilon}|N|)}$, for any constant $\varepsilon > 0$, unless $NP \subseteq ZPP$ \cite{DBLP:journals/jacm/ChuzhoyK09}, we have the following result.

\begin{restatable}{theorem}{teohardOne} 
\label{teo:hard1}
The \StTempSep{$\ell$} problem is 
$2^{\Omega(\log^{1-\varepsilon}|V|)}$-hard to approximate for any constant $\varepsilon> 0$, unless $NP \subseteq ZPP$.
\end{restatable}

The same result holds also for \TempSep{$\ell$},
we need to slightly modify the input temporal graph
so that each $\ell$-temporal path is forced 
to be strict.

\begin{restatable}{corollary}{corhard}
\label{cor:hard1}
The \TempSep{$\ell$} problem is $2^{\Omega(\log^{1-\varepsilon}|V|)}$-hard to approximate for any constant $\varepsilon> 0$, unless $NP \subseteq ZPP$.
\end{restatable}

\subsection{Approximating \StTempSep{$\ell$}}
\label{subsec:ApproxAlgo1}

We consider now the \StTempSep{$\ell$} problem 
and we present an $(\ell-1)$-approximation algorithm
for it.
First, we prove the following
easy claim.
\begin{restatable}{claim}{claimeasy}
\label{claim:easy}
A strict $\ell$-temporal path between $s$
and $z$ contains at most $\ell-1$ vertices different from $s$ and $z$.
\end{restatable}
Now, we present the approximation algorithm 
(Algorithm~\ref{alg:ApproxStrict}).
The algorithm greedily looks for a strict 
$\ell$-temporal path $P$ between $s$ 
and $z$.  
If there exists such a path $P$,  Algorithm~\ref{alg:ApproxStrict} removes all the vertices of $P$ from $G$, except for $s$ and $z$, and adds its vertices to the vertex separator $V'$.
If there exists no such path, 
then the algorithm stops and return $V'$.
\begin{algorithm}
\caption{The approximation algorithm for \StTempSep{$\ell$}}
\label{alg:ApproxStrict}
$V'\gets \emptyset$\;
\While{there exists a strict $\ell$-temporal path $P$,
between $s$ and $z$}
{
    $V' \gets V' \cup (V(P) \setminus \{s,z\})$\;
    Remove $V(P) \setminus \{s,z\}$ from $G$\;
}
\Return($V'$)
\end{algorithm}

Algorithm~\ref{alg:ApproxStrict} has approximation
factor $\ell-1$ since the $\ell$-temporal paths considered
in each iteration are vertex-disjoint, thus any solution
of \StTempSep{$\ell$} contains at least one vertex
for each of these temporal path.
\begin{restatable}{theorem}{TeoStrictApprox}
\label{TeoStrictApprox}
Algorithm \ref{alg:ApproxStrict} is an $(\ell-1)$-approximation algorithm for \StTempSep{$\ell$}.
\end{restatable}
%
%
Next, we prove that, using an approximation preserving reduction
from \UniformHypergraph{}
similar to the one presented in~\cite{DBLP:conf/isaac/HarutyunyanKP23} from
\textsc{Set Cover}, improving this approximation factor
is a challenging problem.
In particular, \UniformHypergraph{}
is not approximable within factor $k$ assuming the Unique Game Conjecture~\cite{khot2008vertex}.

\begin{restatable}{corollary}{hardapprox}
\label{cor:hardapprox}
\StTempSep{$\ell$}  is hard to approximate within factor
$\ell-1$ assuming the Unique Game Conjecture. 
\end{restatable}

\section{Tractability and approximation of \TempCut{$\ell$}}
\label{sec:cut}

In this section we consider the \TempCut{$\ell$}
problem. First, we prove the APX-hardness of
the problem, then we present a $2 \log_2(2\ell)$-approximation algorithm.
\subsection{Hardness of \TempCut{$\ell$}}
We show the APX-hardness of \TempCut{$\ell$}, by presenting an approximation preserving reduction
from the \textsc{Multiway Cut} problem, which is APX-hard~\cite{dahlhaus1992complexity}.
In this latter problem, the input is an undirected graph $G$ with $k \geq 3$ distinguished vertices $v_1, \ldots, v_k$ called \emph{terminals}.  The goal is to remove the minimum number of edges from $G$ to cut all paths between terminals, that is, compute $F \subseteq E(G)$ of minimum size such that in $G - F$ there is no path between $v_i$ and $v_j$ for every distinct $i, j \in [k]$.

Consider an instance $(G, v_1, \ldots, v_k)$ of \textsc{Multiway Cut} and construct a temporal graph $H$ as follows.  First define $\ell = k - 1$.  We assume that some edges of $H$ are \emph{undeletable}.  This can easily be achieved by replacing an undeletable edge $(u, v, t)$ with a large number of parallel paths between $u$ and $v$ at time $t$.
We construct $H$ as follows. First, start with a copy of $G$ in which every temporal edge is defined at time $\ell$, and every such temporal edge is \emph{deletable}: for each $uv \in E(G)$ we add $(u, v, \ell)$ to $H$.  
Then add vertices $s$ and $z$, and add the following \emph{undeletable} edges:
\begin{align*}
&(s, v_1, 1), \quad (v_1, z, \ell + 1); \quad
(s, v_2, 2), \quad (v_2, z, \ell + 2);  
\quad \ldots \\
&(s, v_{k-1}, k - 1), \quad (v_{k-1}, z, \ell + k - 1); 
\quad
(v_k, z, \ell).
\end{align*}
In other words we add $(s, v_i, i)$ and $(v_i, z, \ell + i)$ for $i \in [k-1]$, and add $(v_k, z, \ell)$.  Note that $s$ is a neighbor of every $v_i$ except $v_k$.  
Also observe that the edge $(v_{k-1}, z, \ell + k - 1)$ is useless, but we include it to preserve the pattern.

\begin{restatable}{theorem}
{thmcuthard}
\label{thm:cuthard}
    The \TempCut{$\ell$} problem is APX-hard, even for $\ell = 2$.
\end{restatable}
\noindent 
\emph{Proof sketch.}  Let $F \subseteq E(G)$ be a multiway cut of $G$, so that all $v_i - v_j$ paths are removed in $G - F$.  In $H$, remove the set $F' = \{(u, v, \ell) : uv \in F\}$.  Note that $s$ needs to use two distinct $v_i$ and $v_j$ terminals to reach $z$, because it first needs to use some $(s, v_i, i)$ temporal edges, but cannot use the temporal edge $(v_i, z, i + \ell)$.  But $F'$ cuts any path between two terminals, so $F'$ is a temporal cut of $H$.  

Conversely, any temporal cut $F' \subseteq E(H)$ contains only undeletable edges which are copied from $G$.  Consider $F = \{uv \in E(G) : (u, v, \ell) \in F'\}$.  A case analysis shows that if $G - F$ has a path between $v_i$ and $v_j$, then $H - F'$ has an $\ell$-temporal path from $s$ to $z$, going through $v_i$ then $v_j$ or vice-versa using that path.  Thus $F$ must be a multiway cut for $G$.
\qed

Note for $\ell = 1$, \TempCut{$\ell$} is in P: the subgraphs that occur at different times can be treated independently, and thus $\ell = 1$ can be solved by computing a minimum edge cut for each possible time.

Also note that unlike most of our previous results, there is no obvious extension of the above reduction to the the strict variant, where one must delete a minimum number of edges to remove all strict temporal paths of travel time at most $\ell$.  We leave the complexity of latter as an open problem.

\subsection{A $2 \log_2(2\ell)$-Approximation for \TempCut{$\ell$}}

We present a $2 \log_2(2\ell)$-approximation algorithm for \TempCut{$\ell$}.
We start by recalling a generalization of Menger's Theorem on temporal graphs~\cite{DBLP:journals/networks/Berman96}.
\begin{theorem}
\label{teo:MengersTemporal}
The maximum number of edge-disjoint $s,z$-temporal paths in a temporal graph $G$
equals the size of a minimum temporal $s,z$-cut of $G$.
\end{theorem}
Also given a temporal graph $G$, it is possible
to compute in polynomial time a minimum temporal $s,z$-cut of $G$~\cite{DBLP:journals/networks/Berman96}
and we denote such algorithm by $\mathtt{TempCut}(G)$.
We present now a result that will be useful to prove the approximation factor.

\begin{restatable}{lemma}{lemEdgeShare}
\label{lem:EdgeShare}
Given $t \in [2,\tau -1 ]$, consider temporal graph $G([t,t+\ell-1])$
and let $E'_t$ be an $s, z$-cut of $G([t,t+\ell-1])$. Consider
a temporal path $p$ of $G([t+i, t+\ell+i-1]) - E'_t$, for some $i > 0$,
and a temporal path $p'$ of $G([t-j, t'-j+\ell-1]) - E'_t$, for some $j > 0$.
Then $p$ and $p'$ are temporal edge disjoint.
\end{restatable}

We describe now an algorithm that
has approximation factor 
$\log_2 \tau$, then 
we show how to use it to achieve
approximation factor $2 \log_2(2\ell)$.
We assume for simplicity that $\ell$ is even.
\begin{algorithm}
\caption{$\log_2 \tau$-approximation algorithm for \TempCut{$\ell$}}
\label{alg:ApproxEdge}
\SetKwFunction{TempEdgeCut}{$\mathtt{TempEdgeCut}(G,r,l)$}
$G' \gets G$, \quad $E'\gets \emptyset$, 
\quad $l \gets 1$, $r \gets \tau$\;
Function $\mathtt{\ell-TempEdgeCut}(G,r,l)$\\
\If{ $r-l \geq \ell-1$}
{
    $t \gets \lfloor \frac{r+l}{2} \rfloor$\;
    $E' \gets E' \cup \mathtt{TempCut}(G([t-\frac{\ell}{2}, t + \frac{\ell}{2}-1) ])$\;
    $\mathtt{\ell-TempEdgeCut}(G,l,t+\frac{\ell}{2}-2)$\;
    $\mathtt{\ell-TempEdgeCut}(G,t- \frac{\ell}{2}+1,r)$\;
}
\end{algorithm}
The algorithm is recursive:
given an interval $[l,r]$,
at each step if the difference $r-l$
is at least $\ell-1$, thus $[l,r]$
contains at least $\ell$ timestamps,
it defines a timestamp $t$ that 
partitions in two (almost) equal parts the intervals of length $\ell$ contained in $[l,r]$. 
Then it computes a minimum temporal cut of 
$G[t-\ell/2, t+\ell/2-1]$ 
and recursively
computes a solution on the first part of intervals
and a solution on the second part of intervals (independently).
Since the number of intervals of length $\ell$ is bounded by $\tau - \ell < \tau$, and each recursive call partitions in two equal parts this set of intervals, it follows that after $\log_2 \tau$ 
levels of recursion, the size of an interval is at most
$\ell$, thus there can be at most $\log_2 \tau$ levels of recursion.

Denote by $h$ the number of recursion levels 
and by $w_i$ the number of timestamps defined (as $t$ in the pseudocode) at level $i$, $i \in [h]$.
Denote by $t_{i,j}$ the $j$-th timestamp chosen by the algorithm at the $i$-th level of the recursion, $i \in [h]$ and $j \in [w_i]$.
Now, consider an approximated solution $E'$ and optimal solution $Opt$ of \TempCut{$\ell$}. 
We partition $E'$ 
and $Opt$ as follows.
Given a timestamp $t_{i,j}$, $i \in [h]$ and $j \in [w_i]$,
we denote by $E'(t_{i,j})$ 
the set of temporal edges added to $E'$
by the approximation algorithm when it defines
$t=t_{i,j}$. 
By construction and by Lemma \ref{lem:EdgeShare},
the sets $E'(t_{i,j})$,
$i \in [h]$, $j \in [w_i]$, define a partition of $E'$.
Similarly, we define $Opt(t_{i,j})$,
$i \in [h]$, $j \in [w_i]$, as
the set of temporal edges $e \in Opt$ such that
(1) there is an $\ell$-temporal path in $G([t_{i,j} - \ell/2, t_{i,j}+ \ell/2-1])$ from $s$ to $z$ that contains $e$
and (2) $e$ does not belong to any $Opt(t_{b,q})$, with 
$b < i$ for some $q \in [w_b]$. 
Note that by Lemma~\ref{lem:EdgeShare} this is indeed
a partition, and in particular for each $i\in [h]$ it holds that
$Opt(t_{i,j}) \cap Opt(t_{i,q}) = \emptyset$
for each $j,q \in [w_i]$ and $j \neq q$.

Since each temporal edge $e \in Opt$ belongs to exactly one set $Opt(t_{i,j})$ and  since $Opt_{i,j} \cap Opt_{i,q} = \emptyset$, for each $j, q \in [w_i]$,
we have that
$Opt = \bigcup_{i \in [h],j \in [w_i]} Opt(t_{i,j})$
and $|Opt| = \sum_{i \in [h], j \in [w_i]} Opt(t_{i,j})$.
Now, we prove the following result.

\begin{restatable}{lemma}{lemEdgeShareTwo}
\label{lem:EdgeShare2}
Consider $E'(t_{i,j})$, $i \in [h]$, $j \in [w_i]$,  
the set of temporal edges cut by Algorithm~\ref{alg:ApproxEdge} when it defines $t = t_{i,j}$.
For each level $i \in [h]$ of the recursion, 
it holds that
\[
\sum_{j \in [w_i]} |E'(t_{i,j})| \leq 
\sum_{j \in [w_i]} |Opt(t_{i,j})| + \sum_{b \in [h]:b <i, q \in [w_b]} |Opt(t_{b,q})|.
\]
\end{restatable}
Based on the previous lemma, since the number of recursion levels
is bounded by $\log_2 \tau$, we can prove  the following.

\begin{restatable}{theorem}{logTapprox}  
\label{logTapprox}
Algorithm \ref{alg:ApproxEdge}
returns an $(s,z,\ell)$-temporal cut of size at most 
$\log_2 \tau \cdot  Opt$.
\end{restatable}

Next we prove that the previous approximation  algorithm can be used to obtain an approximation algorithm of factor $2 \log_2(2\ell)$ for \TempCut{$\ell$}.
We assume that $\tau$ is a multiple of $2\ell$ (if not, increase $\tau$ to the next multiple of $2\ell$, with no temporal edges existing in the extra times appended).
Consider the time domain $[1,\tau]$ and computes
the following partitions of intervals of $[1,\tau]$:
\begin{itemize}
\item $P_1 = [1, 2\ell], [2\ell + 1, 4\ell], ..., [\tau - 2\ell+1, \tau]$; \quad
$P_2 = [\ell, 3\ell], [3\ell + 1, 5\ell], ...,  
[\tau - 3 \ell, \tau - \ell]$
\end{itemize}
Each interval of length $\ell$
is contained in one of the interval of the two partitions and moreover two graphs defined on distinct intervals of a set $P_i$ are temporal edge disjoint.
Recall that $G(I)$ is the temporal graph defined on interval $I$.

\begin{restatable}{lemma}{lemapprox} 
\label{lem-approx1}
\label{lem-approx-disjoint}
Consider the sets $P_1$ and $P_2$, then
(1) for each $I=[t,t+\ell-1]$  of length $\ell$, 
there is an interval of $P_1$ or $P_2$
that contains it;
(2) Given $I$ and $I'$ be two distinct interval of $P_i$,
then a temporal path of $G(I)$ and
a temporal path of $G(I')$ are temporal edge disjoint.
\end{restatable}

Now, we run Algorithm \ref{alg:ApproxEdge} on each $G(I)$, with $I \in P_i$ and we compute a feasible solution since by Lemma \ref{lem-approx1} each interval $I$ of length $\ell$ is contained in an interval of $P_1$ or $P_2$.
The $ 2 \log_2(2\ell)$-approximation factor
is due to the fact that each interval $I$ has length $2 \ell$ 
and we approximate each with a factor of $\log_2(2\ell)$.
Then by Lemma \ref{lem-approx-disjoint}
two intervals $I$ and $I'$ of $P_i$, are temporal edge disjoint,
thus each temporal edge removed by an optimal solution belongs
to at most one interval of $P_1$ and at most one interval of $P_2$.  
Because we remove edges from both the intervals in $P_1$ and $P_2$, our approximation factor is multiplied by two.


\begin{restatable}{corollary}{corEdgeApprox}
\label{corEdgeApprox}
\TempCut{$\ell$} admits a $2 \log_2 (2 \ell)$-approximation algorithm.    
\end{restatable}

\medskip

\noindent 
\textbf{Conclusion.}  We briefly conclude with a few open problems: (1) is the \TempSep{$\ell$} problem NP-hard on graphs of treewidth at most three?  
(2) is there a $O(\ell)$-approximation for \TempSep{$\ell$} (non strict variant)?
(3) Is there a $O(1)$-approximation for the \TempCut{$\ell$} problem?  (4) Is the \emph{strict} variant of the \TempCut{$\ell$} NP-hard?

\bibliography{lipics-v2021-sample-article.bib}

\begin{thebibliography}{10}

\bibitem{DBLP:journals/siamdm/BergrenEGK22}
Dan Bergren, Eduard Eiben, Robert Ganian, and Iyad Kanj.
\newblock On covering segments with unit intervals.
\newblock {\em {SIAM} J. Discret. Math.}, 36(2):1200--1230, 2022.
\newblock \href {https://doi.org/10.1137/20M1336412} {\path{doi:10.1137/20M1336412}}.

\bibitem{DBLP:journals/networks/Berman96}
Kenneth~A. Berman.
\newblock Vulnerability of scheduled networks and a generalization of menger's theorem.
\newblock {\em Networks}, 28(3):125--134, 1996.

\bibitem{bodlaender1996efficient}
Hans~L Bodlaender and Ton Kloks.
\newblock Efficient and constructive algorithms for the pathwidth and treewidth of graphs.
\newblock {\em Journal of Algorithms}, 21(2):358--402, 1996.

\bibitem{DBLP:journals/jacm/ChuzhoyK09}
Julia Chuzhoy and Sanjeev Khanna.
\newblock Polynomial flow-cut gaps and hardness of directed cut problems.
\newblock {\em J. {ACM}}, 56(2):6:1--6:28, 2009.
\newblock \href {https://doi.org/10.1145/1502793.1502795} {\path{doi:10.1145/1502793.1502795}}.

\bibitem{dahlhaus1992complexity}
Elias Dahlhaus, David~S Johnson, Christos~H Papadimitriou, Paul~D Seymour, and Mihalis Yannakakis.
\newblock The complexity of multiway cuts.
\newblock In {\em Proceedings of the twenty-fourth annual ACM symposium on Theory of computing}, pages 241--251, 1992.

\bibitem{fellows2009parameterized}
Michael~R Fellows, Danny Hermelin, Frances Rosamond, and St{\'e}phane Vialette.
\newblock On the parameterized complexity of multiple-interval graph problems.
\newblock {\em Theoretical computer science}, 410(1):53--61, 2009.

\bibitem{DBLP:journals/tcs/FluschnikMNRZ20}
Till Fluschnik, Hendrik Molter, Rolf Niedermeier, Malte Renken, and Philipp Zschoche.
\newblock Temporal graph classes: {A} view through temporal separators.
\newblock {\em Theor. Comput. Sci.}, 806:197--218, 2020.
\newblock \href {https://doi.org/10.1016/J.TCS.2019.03.031} {\path{doi:10.1016/J.TCS.2019.03.031}}.

\bibitem{DBLP:conf/isaac/HarutyunyanKP23}
Hovhannes~A. Harutyunyan, Kamran Koupayi, and Denis Pankratov.
\newblock Temporal separators with deadlines.
\newblock In Satoru Iwata and Naonori Kakimura, editors, {\em 34th International Symposium on Algorithms and Computation, {ISAAC} 2023, December 3-6, 2023, Kyoto, Japan}, volume 283 of {\em LIPIcs}, pages 38:1--38:19. Schloss Dagstuhl - Leibniz-Zentrum f{\"{u}}r Informatik, 2023.
\newblock URL: \url{https://doi.org/10.4230/LIPIcs.ISAAC.2023.38}, \href {https://doi.org/10.4230/LIPICS.ISAAC.2023.38} {\path{doi:10.4230/LIPICS.ISAAC.2023.38}}.

\bibitem{holme2015modern}
Petter Holme.
\newblock Modern temporal network theory: a colloquium.
\newblock {\em The European Physical Journal B}, 88(9):234, 2015.

\bibitem{DBLP:journals/jcss/IbiapinaS24}
Allen Ibiapina and Ana Silva.
\newblock Mengerian graphs: Characterization and recognition.
\newblock {\em J. Comput. Syst. Sci.}, 139:103467, 2024.
\newblock URL: \url{https://doi.org/10.1016/j.jcss.2023.103467}, \href {https://doi.org/10.1016/J.JCSS.2023.103467} {\path{doi:10.1016/J.JCSS.2023.103467}}.

\bibitem{DBLP:journals/jcss/KempeKK02}
David Kempe, Jon~M. Kleinberg, and Amit Kumar.
\newblock Connectivity and inference problems for temporal networks.
\newblock {\em J. Comput. Syst. Sci.}, 64(4):820--842, 2002.
\newblock \href {https://doi.org/10.1006/jcss.2002.1829} {\path{doi:10.1006/jcss.2002.1829}}.

\bibitem{khot2008vertex}
Subhash Khot and Oded Regev.
\newblock Vertex cover might be hard to approximate to within 2- $\varepsilon$.
\newblock {\em Journal of Computer and System Sciences}, 74(3):335--349, 2008.

\bibitem{DBLP:journals/aamas/KlobasMMNZ23}
Nina Klobas, George~B. Mertzios, Hendrik Molter, Rolf Niedermeier, and Philipp Zschoche.
\newblock Interference-free walks in time: temporally disjoint paths.
\newblock {\em Auton. Agents Multi Agent Syst.}, 37(1):1, 2023.
\newblock \href {https://doi.org/10.1007/S10458-022-09583-5} {\path{doi:10.1007/S10458-022-09583-5}}.

\bibitem{DBLP:journals/jcss/MaackMNR23}
Nicolas Maack, Hendrik Molter, Rolf Niedermeier, and Malte Renken.
\newblock On finding separators in temporal split and permutation graphs.
\newblock {\em J. Comput. Syst. Sci.}, 135:1--14, 2023.
\newblock \href {https://doi.org/10.1016/J.JCSS.2023.01.004} {\path{doi:10.1016/J.JCSS.2023.01.004}}.

\bibitem{DBLP:journals/im/Michail16}
Othon Michail.
\newblock An introduction to temporal graphs: An algorithmic perspective.
\newblock {\em Internet Math.}, 12(4):239--280, 2016.
\newblock \href {https://doi.org/10.1080/15427951.2016.1177801} {\path{doi:10.1080/15427951.2016.1177801}}.

\bibitem{proskurowski1999classes}
Andrzej Proskurowski and Jan~Arne Telle.
\newblock Classes of graphs with restricted interval models.
\newblock {\em Discrete Mathematics \& Theoretical Computer Science}, 3, 1999.

\bibitem{DBLP:journals/jcss/ZschocheFMN20}
Philipp Zschoche, Till Fluschnik, Hendrik Molter, and Rolf Niedermeier.
\newblock The complexity of finding small separators in temporal graphs.
\newblock {\em J. Comput. Syst. Sci.}, 107:72--92, 2020.
\newblock URL: \url{https://doi.org/10.1016/j.jcss.2019.07.006}, \href {https://doi.org/10.1016/J.JCSS.2019.07.006} {\path{doi:10.1016/J.JCSS.2019.07.006}}.

\end{thebibliography}

\newpage

\section*{Appendix}

\subsection*{Proof of Lemma \ref{lemboundedpathwidth}}

\lemboundedpathwidth*

\begin{proof}
    For the pathwidth, we describe an explicit sequence of bags that form the path decomposition.  
    First, start with a bag containing $\{s, z, r\}$.  These three vertices will be in every subsequent bag.  

    Then, for $i = 1, 2, \ldots, k$ in order, handle the vertices of the $V_i$ gadget as follows.  We assume inductively that before handling the $i$-th gadget, the last bag in the sequence is $\{s, z, r\}$ (which will also hold after handling the gadget).
    Let $V_i = \{u_1, \ldots, u_n\}$, ordered as in the construction.
    Add to the sequence the bags
    \begin{align*}
    & \{s, z, r, \uout{u}_1\} & \mbox{(introduce $\uout{u}_1$)} \\
    & \{s, z, r, \uout{u}_1, p_i\} & \mbox{(introduce $p_i$)} \\
    & \{s, z, r, \uout{u}_1\} & \mbox{(forget $p_i$)} \\
    & \{s, z, r, \uout{u}_1, \uin{u}_1\} & \mbox{(introduce $\uin{u}_1$)} \\
    & \{s, z, r, \uout{u}_1\} & \mbox{(forget $\uin{u}_1$)} \\
    & \{s, z, r, \uout{u}_1, \uout{u}_{2}\} & \mbox{(introduce $\uout{u}_{2}$)} \\
    & \{s, z, r, \uout{u}_{2}\} & \mbox{(forget $\uout{u}_{1}$)} \\
    & \{s, z, r, \uout{u}_{2}, \uin{u}_{2}\} & \mbox{(introduce $\uin{u}_{2}$)} \\
    & \ldots \\
    & \{s, z, r, \uout{u}_n, \uin{u}_n\} & \mbox{(introduce $\uin{u}_n$)} \\
    & \{s, z, r, \uout{u}_n\} & \mbox{(forget $\uin{u}_n$)} \\
    & \{s, z, r\} & \mbox{(forget $\uout{u}_n$)} \\
    \end{align*}
In other words for $h \in [n - 1]$, the pattern is to introduce $\uout{u}_h$, then introduce $\uin{u}_h$, forget it, introduce $\uout{u}_{h+1}$, forget $\uout{u}_h$, and repeat the process (with a special case for $p_i$ initially).
It is clear that the bags containing a $\uout{u}_h$ or $\uin{u}_h$ vertex form a connected subgraph, as they are introduced and forgotten once.  To see that the edges of $H$ incident to a gadget vertex all occur in some bag, note that the vertices adjacent to a $\uin{u}_h$ vertex are $\uout{u}_h, s, z$, and $r$.  These are all part of a common bag with $\uin{u}_h$ at some point, so these adjacencies are covered.
The vertices adjacent to a $\uout{u}_h$ vertex are among $\uin{u}_h, s, z$, $\uout{u}_{h-1}$ if $h > 1$, and $\uout{u}_{h+1}$ if $h < n$, which can all be seen to be covered by the sequence (also note that $p_i$ is correctly handled as well).

Then, again assuming we have reached a $\{s, r, z\}$ bag, we handle each edge $e = \{u, v\}$ of $G$ in an arbitrary order.  This can simply be done by adding to the sequence of bags 
\begin{align*}
    & \{s, z, r, w_u^e, w_v^e\} & \mbox{(introduce $w_u^e, w_v^e$)} \\
    & \{s, z, r, w_v^e\} & \mbox{(forget $w_u^e$)} \\
    & \{s, z, r, w_u^e, q_e\} & \mbox{(introduce $q_e$)} \\
    & \{s, z, r\} & \mbox{(forget  $w_v^e, q_e$)}
    \end{align*}
    One can easily check that the adjacencies of $w_u^e$, $w_v^e$, and $q_e$ are covered by this sequence of bags.

    The largest bag has $5$ elements, and so the pathwidth of $H$ is at most $4$.
\end{proof}


\subsection{Proof of Lemma~\ref{lem:pwhard}}

\lempwhard*

\begin{proof}
Suppose that $G$ admits a multicolored independent set $I = \{u_1, \ldots, u_k\}$ where, for each $i \in [k]$, $u_i$ is in $V_i$.  

We first observe that in $H$, there are $t_{green}$ temporal paths that traverse $s$, $u^-$, $ u^+$ and $z$ and $s$, $w_u^e$, $w_v^e$ and $z$.  All these temporal paths are vertex-disjoint and it suffices it remove one vertex from each to destroy all $t_{green}$ paths.  The same reasoning applies to $t_{red}$ temporal paths: it suffices to remove one of $p_i, u_1^-, \ldots, u_n^-$ to destroy all such paths.

In $H$, for every $i \in [k]$, delete $\uout{u}_i$
(recall $u_i \in I$).  Then for every $u \in V_i \setminus \{u_i\}$, delete $\uin{u}$ instead.  Therefore, for each $i \in [k]$, $\uin{u}_i$ is the only ``plus'' vertex kept.
At this point, for $u \in V(G)$, any $t_{green}$ temporal path from $s$ to $z$ that contains a $\uin{u}$ or $\uout{u}$ vertex is destroyed, since we deleted one of the two.
Moreover, because we deleted  $\uout{u}_i$ for each $i \in [k]$, any temporal $t_{red}$ path is destroyed.
As for the $t_u$ temporal paths, we note for later use that 
there is a temporal path from $s$ to $r$ with time $t_u$ that goes through at least one vertex $u$ of a $V_i$ gadget  if and only if $u = u_i \in I$. This is because any such path needs to go through a ``plus'' vertex, and those are only kept for $\uin{u}_i$ vertices with $i \in [k]$. 

Next, consider an edge $e = uv$ of $G$, where $u \in V_i$ and $v \in V_j$ with $i < j$.  
Suppose first that $v \notin I$.  Then we delete $w_u^e$.
This destroys the $t_{green}$ temporal path that traverses
$s$, $w_u^e$, $ w_v^e$ and $z$.
Note that there is the temporal subpath that traverses 
$r$, $w_v^e$, $q_e$, $z$ at time $t_v$.
To see that no $s - z$ temporal path at time $t_v$ involving $w_v^e$ is possible, observe that such a temporal path would need to take the temporal edge $(s, v^+, t_v)$, which is the only edge at time $t_v$ incident to $s$.  However we have removed $v^+$ since $v \notin I$, and therefore there is no such temporal path.
Next, if $v \in I$, then $v = u_i$ for some $i \in [k]$.  Because $I$ is an independent set, we have $u \notin I$.  In this case, delete $w_v^e$.  This also destroys the $t_{green}$ temporal path going through $w_u^e$ and $w_v^e$.  By the same reasoning as in the previous case, $u \notin I$ implies that we deleted $u^+$ and that no temporal path with time $t_u$ goes through $w_u^e$.

It follows that the set of deleted vertices remove all temporal paths of travel time $\ell = 1$.
For each $u \in V(G)$, we have deleted exactly one of $\uout{u}$ or $\uin{u}$.  Moreover, for each $e \in E(G)$, we have deleted exactly one of $w_u^e$ or $w_v^e$.  Thus the number of deletions is $|V(G)| + |E(G)|$, as desired.

Conversely, suppose that all temporal paths of travel time $1$ can be destroyed with at most $|V(G)| + |V(E)|$ vertex deletions.
Note that for each $u \in V(G)$, one of $\uout{u}$ or $\uin{u}$ must be deleted to destroy the $t_{green}$ temporal path. 
Likewise, for each edge $e = \{u, v\}$ of $G$, one of $w_u^e$ or $w_v^e$ must be deleted to destroy the $t_{green}$ temporal path.
These amount to at least $|V(G)| + |V(E)|$ vertex deletions, so it follows that we cannot delete both $\uout{u}$ and $\uin{u}$ in the first case, or both $w_u^e$ and $w_v^e$ in the second case, so there is exactly one deletion per pair.  Importantly, we also note that $r$ cannot be deleted and that no $p_i$ or $q_e$ vertex can be deleted.

Next, for $i \in [k]$, notice that because of the $t_{red}$ temporal path, there must be at least one $u \in V_i$ such that $\uout{u}$ is deleted, and therefore such that $\uin{u}$ is not deleted.  
For $i \in [k]$, let $u_i$ be the vertex of $V_i$ such that $\uin{u}_i$ is not deleted (if multiple choices arise, choose arbitrarily).

We claim that $I = \{u_1, \ldots, u_k\}$ is an independent set of $G$.
Assume instead that $e = u_i u_j$ is an edge of $G$, for some $i < j$, so that $u_i^+$ and $u_j^+$ are in the remaining temporal graph.
Consider the $w_{u_i}^e$ and $w_{u_j}^e$ vertices, one of which is deleted and the other is not.
If $w_{u_i}^e$ is not deleted, then there is a $t_{u_i}$ temporal path traversing $s$, $ \uin{u}_i$, $r$ $w_{u_i}^e$ and $z$.  
If $w_{u_j}^e$ is not deleted, then there is a $t_{u_j}$ temporal path that traverses $s$, $\uin{u}_j$, $r$, $w_{u_j}^e$, $q_e$ and $z$.  Both cases lead to a contradiction, so $e$ cannot exist, and therefore $I$ is a multicolored independent set.
\end{proof}


\subsection*{Proof of Lemma \ref{lem:pwstructure}}

\pwstrcuture*

\begin{proof}
	Let $X_1, X_2, \ldots, X_p$ be the sequence of bags of any nice path decomposition of $G$, where each bags has at most $4$ vertices. 
	We claim that some bag contains both $s$ and $z$.  Since vertices are introduced one at a time, one of $s$ or $z$ is introduced before the other.  Assume that $s$ is introduced in some bag $X_p$, before $z$ is introduced, and let $X_q$ be the bag in which $s$ is forgotten, with $p < q$.  If no bag contains both $s$ and $z$, then $z$ is introduced in $X_{q+1}$ or later, meaning that $X_q$ is an $(s, z)$-separator of size $|X_q| \leq 3$, a contradiction.  So $s$ and $z$ are together at some point.

    Assume that the sequence $X_1, \ldots, X_p$ is chosen so that, among the possibilities, the first bag $X_c$ that contains both $s$ and $z$ has maximum size and, under this condition, the last bag $X_d$ that contains both $s$ and $z$ has maximum size.
    Notice that $[c, d] = \{c, c+1, \ldots, d\}$ is the interval of bags that contain both $s$ and $z$.  
    Also note that one of $s$ or $z$ is introduced in $X_c$ (and was thus absent from $X_{c-1}$) and one is forgotten in $X_{d+1}$.  In particular, $X_{c-1}$ and $X_{d+1}$ exist.  
 
    We show that$|X_c| = |X_d| = 4$.  
    Suppose that $|X_c| \leq 3$, and note that $X_{c-1}$ has one less vertex.  
    Let $X_e$ be the first introduce bag that occurs after $X_c$, i.e., $e > c$ and is the minimum value such that $X_e$ introduces a vertex, say $x$.  If no such $X_e$ exists, then the vertex $s$ or $z$ introduced in $X_c$ has only members of $X_c$ in its neighborhood, because no further vertex is introduced,  and so there would be an $(s, z)$-separator of size at most $|X_c \setminus \{s, z\}| \leq 1$.
    So $X_e$ exists, and in that case we can simply introduce $x$ earlier.  That is, modify the sequence of bags by adding, after $X_{c-1}$ and before $X_c$, a new bag $X_{c-1} \cup \{x\}$ that introduces $x$, and then adding $x$ to each bag $X_{c}, X_{c+1}, \ldots, X_{e-1}$.  Then remove $X_e$ since it has become equal to $X_{e-1}$.  One can easily see that this results in another nice path decomposition, and the fact that $|X_c| \leq 3$ and that no vertex is introduced until $X_e$ implies that the bag sizes remain bounded by $4$.  
    Moreover, this increases the size of the first bag that has both $s$ and $z$, contradicting our choice of bag sequence.  So we have $|X_c| = 4$. 
    
    Notice that in the above argument, modifying the sequence of bags does not decrease the size of the last bag that contains both $s$ and $z$.  We may thus assume that $|X_c| = 4$, reverse the sequence of bags, apply the argument to make the first bag with $s$ and $z$ in this reversed sequence also have size $4$ (without affecting the last such bag), resulting in $|X_c| = |X_d| = 4$.

    Next, we argue that $d > c$.
	Let us suppose that $c = d$.  If $X_c$ introduces $s$ and $X_{c+1}$ forgets $s$, then 
	$s$ has at most three neighbors and they are all in $X_c$, in which case an $(s, z)$-separator of size at most $|X_c \setminus \{s, z\}| = 2$ would exist.
	The same holds if $X_c$ introduces $z$ and $X_{c+1}$ forgets it.  
	So assume that $X_c$ introduces $z$ and $X_{c+1}$ forgets $s$.
	In that case, the properties of path decompositions imply that $X_c \setminus \{s, z\}$ must be an $(s, z)$-separator, a contradiction (this is because a path from $s$ to $z$ must start with a vertex in $X_1, \ldots, X_c$ and end with a vertex in $X_c, \ldots, X_p$, and must thus use a vertex of $X_c$).  We deduce that $d > c$.

    We can finally define $A$ and $C$.  
    Since we know that $|X_c| = 4$, denote $X_c = \{s, z, x, y\}$ for some vertices $x, y$.  Because $d > c$, we have $s, z \in X_{c+1}$, and so the bag $X_{c+1}$ cannot introduce a new vertex. Thus it forgets $x$ or $y$, let us say $y$.
    In that case, define $A = (X_1 \cup \ldots \cup X_c) \setminus \{s, z, x\}$ and put $w_1 = x$.  Notice that $A \neq \emptyset$ and $y$ is the only vertex of $A$ that can be a neighbor of the vertex $s$ or $z$ that is introduced in $X_c$.  This justifies the first part of our statement regarding $A$.  Also note that no vertex introduced later than $X_c$ can have a neighbor in $A$, and so $x = w_1$ is the only vertex outside of $A \cup \{s, z\}$ that can have neighbors in $A$ (in particular, $y$ is not problematic since it is forgotten in $X_{c+1}$).  
 
    To define $C$, let $X_{d} = \{s, z, x', y'\}$. 
    Then $X_{d}$ cannot be a forget bag, and must thus introduce one of $x'$ or $y'$, say $y'$.  Define $C = (X_d \cup X_{d+1} \cup \ldots \cup X_p) \setminus \{s, z, x'\}$ and $w_k = x'$.  Notice that $y'$ is the only vertex of $C$ that can be a neighbor of the vertex forgotten in $X_{d+1}$, which is either $s$ or $z$.  
    This justifies the first part of the statement regarding $C$.  Also note that $x' = w_k$ is the only vertex outside of $C \cup \{s, z\}$ that can have neighbors in $C$.
	
    Finally, let $B = V(G) \setminus (A \cup C \cup \{s, z\})$.
    The previous arguments imply that our statement on the neighborhood of $w_1$ and $w_k$ holds, and it remains to argue that $B$ induces a caterpillar whose main path starts at $w_1$ and ends at $w_k$.
    To see this, we note that a path decomposition of $G[B]$ can be obtained by the sequence of bags $X_c \setminus \{s, z\}, \ldots, X_d \setminus \{s, z\}$.  Each bag has at most two elements, and thus $G[B]$ has pathwidth $1$.  It is well-known~\cite{proskurowski1999classes} that such graphs contain only connected components that are caterpillars.  We note that $G[B]$ must be connected, as otherwise $G - \{s, z\}$ cannot be connected, whereas we assume that $G$ is clean.  Thus $G[B]$ is a caterpillar, and it can be seen that the ends of the main path can be defined as $w_1$ and $w_k$ since the first and last bag of the sequence only contain $w_1$ and $w_k$, respectively (in the caterpillar, $w_1$ could be either the end of the main path, or a leaf adjacent to an end of the main path --- in the latter case we can simply extend the main path by adding $w_1$ to it; the same applies to $w_k$).  
\end{proof}

\subsection*{Proof of Lemma \ref{lem:pwdels}}

\pwdels*

\begin{proof}
    Let $D \subseteq V(G)$ be a minimum $(s, z, \ell)$-temporal separator.  
    Suppose $D$ contains at least two vertices $u, u' \in A$.  Let $v$ be a vertex of $A$ such that $N(s) \cap A \subseteq \{v\}$ or $N(z) \cap A \subseteq \{v\}$, which is guaranteed to exist by Lemma~\ref{lem:pwstructure}. 
    Aside from $s, z$, recall that $w_1$ is the only vertex outside of $A$ to have neighbors in $A$.  Therefore,  all $s-z$ paths of $G$ that contain a vertex of $A$ must go through $v$ or $w_1$, whether these paths are temporal or not.  In particular, any path that goes through $u$ or $u'$ also goes through $v$ or $w_1$.  Thus $(D \setminus \{u, u'\}) \cup \{v, w_1\}$ also cuts such a path.  We may thus assume that the two deleted vertices are $v$ and $w_1$ instead of $u$ and $u'$, and thus there is a single deletion in $A$.
    Under that assumption, notice that we may apply the same argument if two vertices of $C$ are deleted, and so we may assume a single deletion or less in each of $A$ and $C$.

    Next, suppose that some element $y$ of $B$ is in $D$, but $y$ is not on the main path of the caterpillar.  
    Then $N(y) \subseteq \{s, z, w_i\}$ for some $w_i$ of the main path. 
    Because $G$ is clean, $[(s, y,t), (y,z,t')$ is not an $(s, z, \ell)$-temporal path, and thus any $s - z$ path that goes through $y$ must also go through $w_i$.  If $w_i \in D$, then we can just remove $y$ from $D$.  If instead $w_i \notin D$, then $(D \setminus \{y\}) \cup \{w_i\}$ cuts all the $s - z$ paths that contain $y$.  Since this does not alter the intersection of $D$ with $A$, with $C$, or with other vertices of the main path, we can apply this reasoning to each $y$ of $D \cap B$ not on the main path, and satisfy at once all the requirements of the lemma on the intersection of $D$ with $A, B, $ and $C$.
\end{proof}


\subsection*{Proof of Theorem~\ref{thm:pw3}}

\thmpwthree*

\begin{proof}
	Let us argue the correctness of the algorithm.  We assume a pre-processing step that ensures that we only run the algorithm on a clean graph $G$.  

    Let us first argue that when the routine $extendSeparator$ does not return ``impossible'', then the returned set $D$ is an $(s, z, \ell)$-temporal separator.
    For clarity, denote by $D_0$ the set $D$ initially received as input, then by $D_i$ the contents of the set $D$ after the $i$-th iteration of the main loop of the routine, for $i = 1, 2, \ldots, k-1$, and finally denote by $D_k$ the set returned by the algorithm.
    
    By the initial check made by $extendSeparator$, we know that $G - (D_0 \cup B)$ has no $(s, z, \ell)$-temporal path.  Thus any $(s, z, \ell)$-temporal path must use vertices of $D_0  \cup B$.
    One can then see inductively that there is no $(s, z, \ell)$-temporal path in $G[W_i] - D_i$ for each $i = 1, \ldots, k - 1$.  Indeed, for $i = 1$, if we don't add $w_1$ to $D_1$, this is because $G[W_1] - D_0$ already had no temporal path, and if we add $w_1$ to $D_1$, a remaining temporal path of $G[W_1] - D_1$ would only involve $A \setminus D_0$, which we excluded at the start of the routine (that temporal path cannot go through leaf neighbors of $w_1$ in the caterpillar, as we cleaned temporal paths of length $2$ going through them, and they are not connected to $A$).   
    For larger $i$, assuming that $G[W_{i-1}] - D_{i-1}$ has no temporal path, we get that any temporal path of $G[W_i] - D_{i-1}$ must include $w_i$; but if this happens we include $w_i$ in $D_i$.  Therefore, our assumption holds after every iteration of the main loop.
    After the main loop, if any $(s, z, \ell)$-temporal path remains in $G - D_{k-1}$, it must include $w_k$ (if the temporal path uses vertices of $C$, that temporal path must use $B$ and thus go through $w_k$ in that temporal path, and if the temporal path does not use $C$, we know that $G[W_{k-1}] - D_{k-1}$ has no temporal path, so it must use $w_k$).  If that happens, we put $w_k$ in $D_k$, cutting all possible temporal paths.  Thus when $D_k$ is returned it must be an $(s, z, \ell)$-temporal separator.  It follows that the set returned by $getTemporalSeparator$ is indeed an $(s, z, \ell)$-separator and its size is at most the size of a minimum  $(s, z, \ell)$-separator (unless every call to $extendSeparator$ returns ``impossible'', which we next exclude).

    We next show that $getTemporalSeparator$ returns a solution of minimum size.
    Let $D^*$ be a minimum $(s, z, \ell)$-separator of $G$, chosen so that there is at most one deletion 
	in each of $A$ and $C$, and such that $D^* \cap B \subseteq \{w_1, \ldots, w_k\}$ (such a $D^*$ exists by Lemma~\ref{lem:pwdels}).
	Moreover, refine the choice of $D^*$ among the possibilities as follows: let $\{w_{a_1}, w_{a_2}, \ldots, w_{a_p}\} = D^* \cap B$, 
	where $a_1 < a_2 < \ldots < a_p$, and choose $D^*$ such that the vector $(a_1, a_2, \ldots, a_p)$ is lexicographically maximum (that is, the first vertex $w_{a_1}$ to delete is as late as possible, then among the choices, $w_{a_2}$ is as late as possible, and so on).

    Consider the iteration of $getTemporalSeparator$ that guesses the same deletions $D \in \mathcal{D}$ in $A$ and $C$ as $D^*$, so that  $D \cap A = D^* \cap A$ and $D \cap C = D^* \cap C$.  This will happen since we try every possibility.
    When $D$ is passed as input to $extendSeparator$, because $D^*$ contains only elements of $D$ and $B$, we have that $G - (D \cup B)$ has no $(s, z, \ell)$-temporal path, so when we call $extendSeparator$ with $D$ as input, it will not return ``impossible''.
    Again denote $D_0 = D$ as the input to $extendSeparator$, and $D_i$ as the $D$ set after the $i$-th iteration of the main loop, and $D_k$ the returned set.

   Note that $D_k$ is obtained by only adding the $w_i$ vertices to $D$, so we have $D_k \cap A = D^* \cap A$ and $D_k \cap C = D^* \cap C$.     
   We just need to argue that $D_k \cap B = D^* \cap B$.
	If $D^* \cap B = \emptyset$, then there is no $(s, z, \ell)$-temporal path in $G - D$, and thus no such path in any $G[W_i] - D_0$.  Hence our algorithm will never delete any element of $B$ and returns the correct deletion set.
	So assume that $D^* \cap B \neq \emptyset$. 
	Let $w_{b_1}, w_{b_2}, \ldots, w_{b_q}$ be the vertices in $D_k \cap B$, with $b_1 < b_2 < \ldots < b_q$.
	Note that $q \geq p$, by the optimality of $D^*$ and our previous argument that $D_k$ is an $(s, z, \ell)$-temporal separator.
	We claim that the $w_{a_i}$'s and the $w_{b_i}$'s form the same set of vertices.
	Assume otherwise, and look at the first difference, that is, assume that there is some $i \in [q]$ such that $i > p$, or such that $w_{b_i} \neq w_{a_i}$, and let $i$ be the smallest such index.
	
    Suppose first that $i > p$, so that $i = p+1$.  Then $w_{a_1} = w_{b_1}, \ldots, w_{a_p} = w_{b_p}$, meaning that our algorithm first deletes the same set of vertices as $D^*$, then adds an additional vertex $w_{b_p+1}$.  This can only happen if, after the first $p$ deletions, the algorithm still finds an $(s, z, \ell)$-temporal path in $G[W_{b_p+1}] - D_{b_p}$ (or $G - D_{b_p + 1}$ if this happens after the $for$ loop).  This would imply that $D^*$ does not destroy every $(s, z, \ell)$-temporal path, a contradiction.  
	
	So $i \leq p$ must occur and $a_i \neq b_i$.  Suppose that $a_i > b_i$.  This means that the algorithm and $D^*$ make the same first $i - 1$ deletions on the main path, and that there is an $(s, z, \ell)$-temporal path in $G[W_{b_i}] - D_{b_i - 1}$. 
    But $a_i > b_i$ implies that $D^*$ does not remove any vertex of $W_{b_i}$ other than $D_{b_i - 1}$, and so it contains that $(s, z, \ell)$-temporal path, a contradiction.
    
    Finally, suppose that $a_i < b_i$. In $D^*$, we can replace $w_{a_i}$ with $w_{a_i + 1}$ and get another solution: the algorithm did not add $w_{a_i}$, so there is no $(s, z, \ell)$-temporal path in $G[W_{a_i}] - D_{a_i - 1}$.  Thus,  all $(s, z, \ell)$-temporal paths of $G[W_{a_i}] - D_{a_i - 1}$ that use $w_{a_i}$ also use $w_{a_i + 1}$, so the latter can be used to cut all paths that $w_{a_i}$ cuts. 
    After doing this replacement, we obtain an alternate solution with higher lexicographical order, contradicting the choice of $D^*$.
    
    It follows that we may assume $D_k = D^*$.  Since the main algorithm $getTemporalSeparator$ returns that $D_k$ or better, combined with the argued fact that it always returns some $(s, z, \ell)$-temporal separator, we deduce that the algorithm is correct.

    \medskip 

    \noindent
    \textbf{Complexity.}  Recall that there is an implicit cleaning step, which requires checking for an $(s, z)$-separator of size $3$ or less.  For each of the $O(n^3)$ subsets of three vertices or less, we spend time $O(n + m)$ to check for connectedness, so $O(n^4 + n^3 m)$.  We can also cut all $(s, z, \ell)$-temporal paths of length $2$ without increasing this complexity, as well as computing the connected components of $G - \{s, z\}$.  We may assume that $G = \{s, z\}$ is connected, as the sum of  times for handling the components separately is not greater than if $G - \{s, z\}$ is connected.  Also note that a nice path decomposition of width $3$, if it exists, can be built in linear time~\cite{bodlaender1996efficient}.

    Then, $getTemporalSeparator$ makes $O(n^2)$ calls to $extendSeparator$.  The complexity of the latter is dominated by the time needed to check for the existence of an $(s, z, \ell)$-temporal path, up to $|B| \in O(n)$ times.  Each check can be done in time $O(nm)$, see~\cite[Lemma 2.1]{DBLP:conf/isaac/HarutyunyanKP23}, and so the routine takes time $O(n^2 m)$.  Multiplying by the number of calls, we get a total time of $(n^4m)$.  
\end{proof}


\subsection*{Proof of Lemma \ref{lem:DMCSt}}

\lemDMCST*
\begin{proof}
(1)
Consider a solution $A'$ of \DMC{} on instance $(D,R)$ consisting of $k$ arcs.
Then we define in polynomial time a solution of 
\StTempSep{$\ell$} on the corresponding instance 
$(G,s,t, \ell)$ as follows:
$V' = \{ y_{i,j}: (v_i,v_j) \in A'\}$.
Clearly $|V'| = |A'| = k$.
Next, we show that $V'$ is a solution of \StTempSep{$\ell$}, that is there is no strict temporal path from $s$ to $z$ of travelling time at most $\ell$
in $G$.
Aiming at a contradiction, assume that in 
$G[V \setminus V']$ there exists a strict temporal path 
$P= [(s,w_1, t_1),  (w_1, w_2, t_2), \dots, (w_q, z,t_q)] $,
with $t_q - t_1 +1 \leq \ell$, from $s$ to $z$,
and note that $q \leq \ell$, since $P$ is strict. 

First, we prove some properties of $P$:
\begin{enumerate}
    \item By construction $w_1$ must be a vertex $x_{i,q}$, for some $i \in [|N|]$ and $q \in [|A|+1]$, with $v_i \in R_S$.

    \item The temporal edge
    from $s$ to $x_{i,q}$ is the following:
    $(s,x_{i,q},\ell \cdot (2i-2))$; hence
    $t_1 = \ell \cdot (2i-2)$.
    
    \item By construction $w_q$ must be a vertex $x_{j,r}$, for some $j \in [|N|]$ and 
    $r \in [|A|+1]$, with $v_j \in R_Z$
    and $v_j \in Reach(D,v_i)$.


    \item Each temporal edge of $P$, different from the first one and the last one, is of the form 
    $(x_{a,q}, y_{a,b}, t)$ or $(y_{a,b}, x_{b,r}, t)$,
    with $a,b \in [|N|]$, $q,r \in [|A|+1]$, 
    $t \in [\ell \cdot (2i-2),\ell \cdot (2i-1)-1]$.

\item By construction $y_{a,b}$ is adjacent
    only to the set of vertices 
    $\{x_{a,q}: q \in [|A|+1]\}$ and 
    $\{x_{b,r}: r \in [|A|+1]\}$;
    thus  
    if $P$ contains a temporal
    edge $(x_{a,q}, y_{a,b}, t)$, 
    then it must contain two consecutive 
    temporal edges $(x_{a,q}, y_{a,b}, t)$, 
    $(y_{a,b}, x_{j,r}, t')$, with
    $t < t'$ and 
    $t,t' \in [\ell \cdot (2i-2),\ell \cdot (2i-1)-1]$.
    Note that $a < b$ and that in particular
    $a \neq b$, since $t < t'$.
    In particular if $a = b$, then
    $t \in [\ell \cdot (2i-2),\ell \cdot (2i-1)-1]$
    and  $t' \in [\ell \cdot (2j-2),\ell \cdot (2j-1)-1]$, with $i < j$. But then  
    $t' \geq \ell \cdot (2(i+1)-2)$ and $t \leq \ell \cdot (2i-1)-1$, thus $t' -t  > \ell$.  
    
\item If $P$ contains  two consecutive 
    temporal edges $(x_{a,q},y_{a,b},t)$, 
    $(y_{a,b},x_{b,r}, t')$, 
    with $t < t'$ and   $t,t' \in [\ell \cdot (2i-2),\ell \cdot (2i-1)-1]$,     
    it follows that
    there exists an arc $(v_a,v_b) \in A$.
    
\end{enumerate}

Consider the vertices of $D$,
such that $x_{a,q}$ is in $P$.
We assume w.l.o.g. that these vertices are 
$v_i, v_{i,1}, \dots , v_j$.
By the above properties, in particular Points~4,~5,~6, it follows that the vertices $v_i, v_{i,1}, \dots , v_j$ induce  
a directed path in $D -  A'$.
Moreover, by Points~1,~2 it follows that 
$v_i \in R_S$ 
and by Point~3 $v_j \in R_Z$. Now,
$P$ is a path of travelling time at most $\ell$, since
$t_q - t_1 +1 \leq \ell$. 
By definition of temporal
edge $(x_{j,r}, z, t)$,
it follows
that $t = \ell \cdot (2i-1)-1 $ and thus that $(v_i, v_j) \in R$.
We conclude that $D - A'$ contains a pair of
terminals that are not disconnected, contradicting the assumption that $A'$ is a solution of \DMC{}.
Hence after the removal of the set $V'$, $G$ does not
contain any strict temporal path of traveling time at most $\ell$ from $s$ to $t$,
thus it is a solution of \StTempSep{$\ell$}
on instance $(G,s,z, \ell)$.

(2) Consider a solution $V'$ of \StTempSep{$\ell$} 
on instance $(G,\ell)$, where $|V'|= k$.
First, we show some properties of $V'$.
Note that, for each $i \in [|N|]$, there exist
$|A|+1$ vertices $x_{i,q}$, $q \in [|A|+1]$.
Each two vertices $x_{i,q}$, $x_{i,r}$ are adjacent to the 
same set of vertices of $G$, with temporal 
edges having identical timestamps.
If some vertex $x_{i,q}$ is in $V \setminus V'$, it follows that by construction
there is no strict temporal path of arrival time
at most $\ell$ from $s$ to $z$ that traverses
$x_{i,q}$.
Then  if some vertex $x_{i,r} \in V'$, 
we can remove it from $V'$, 
and $G[(V \setminus V') \cup \{ x_{i,q}: q \in [|A|+1] \}]$ is a temporal graph that does not
contain strict $\ell$-temporal paths from $s$ to $z$,
thus $V' \setminus \{ x_{i,q}: q \in [|A|+1] \}$
is a solution of \StTempSep{$\ell$} 
on instance $(G,s,z, \ell)$,
and clearly $|V' \setminus \{ x_{i,q}: q \in [|A|+1] \}| \leq |V'|$.
Hence it holds that for each $i \in [|N|]$,
either all the vertices $x_{i,q}$, $q \in [|A|+1]$,
belong to $V'$ or no vertex $x_{i,q}$, $q \in [|A|+1]$,
belongs to $V'$.

Now, assume that $\{ x_{i,q}: q \in [|A|+1] \} \subseteq V'$.
It follows that $|V'| \geq |A| +1$. 
By construction the set of vertices 
$\{y_{i,j}: (v_i, v_j) \in A \}$
are $|A|$, since there exists a vertex
$y_{i,j}$ for each arc $(v_i,v_j) \in A$.
Thus we can compute an $(s,z,\ell)$-separator $V^*$ of $G$, such that $|V^*| < |V'|$, as follows:
\[
V^* = \{y_{i,j}: (v_i, v_j) \in A \}.
\]
Clearly $|V^*| < |V'|$, as it contains $|A|$ vertices,
while $V'$ contains at least $|A|+1$ vertices.
Furthermore, the only temporal paths from $s$ to $z$ in $G[V \setminus V^*]$ must have length two and 
pass trough a vertex $x_{i,q}$, for some
$i \in [|N|], q \in [|A|+1]$.
This temporal path consists of  
a temporal edge $(s,x_{i,q},\ell \cdot (2i-2))$
and a temporal edge $(x_{i,q},z,\ell \cdot (2i -1) -1)$,
which implies that $(v_i, v_i) \in R$ and 
this cannot be the case, since we assume 
that $v_i$ has been removed.

It follows that $V'$ is a solution
of \StTempSep{$\ell$} on instance 
$(G,s,z, \ell)$
that
contains only vertices $y_{i,j}$, for some $(v_i,v_j) \in A$.
We define a solution of \DMC{} on instance $(D,R)$ consisting of $k$ arcs as follows:
\[
A' = \{ (v_i,v_j): y_{i,j}  \in V'\}.
\]
Clearly $|V'| = |A'| = k$.

Aiming at a contradiction, assume that $D - A'$ contains a 
path $P$ from some $v_i$ to some $v_j$, with $(v_i, v_j) \in R$.
Define the following temporal path $P'$, corresponding to $P$, in $G[V \setminus V']$:

\begin{enumerate}

\item $P'$ starts with temporal
edge $(s,x_{i,q},\ell \cdot (2i-2))$,
for some $q \in [|A|+1]$; this temporal edge
exists since $v_i \in R_S$.

\item $P'$ ends with temporal
edge $(x_{j,r},z,\ell \cdot (2i-1) -1)$,
for some $r \in [|A|+1]$;
this temporal edge
exists since $v_j \in R_Z$, $v_j \in Reach(D,v_i)$,
indeed $P$ is path from $v_i$ to $v_j$ in $D$,
and $(v_i, v_j) \in R$.
Furthermore, note that
$\ell (2i-1) -1 - \ell(2i - 2) + 1 = \ell$,
hence $tt(P') = \ell$.

\item For each arc $(v_a, v_b) $ of $P$,
$P'$ contains temporal edge
$(x_{a,q},y_{a,b},\ell \cdot (2i-2) +2a-1)$
and 
$(y_{a,b},x_{b,r},\ell \cdot (2i-2) +2(b-1))$;
this temporal edge
exists since $(v_a, v_b) \in A$.
Recall that, since $D$ is a directed acyclic graph,
we assume that for 
each arc $(v_a, v_b) \in A$, we have $ a < b$.
This ensures that the two temporal edges
$(x_{a,q},y_{a,b},\ell \cdot (2i-2) +2a-1)$
and 
$(y_{a,b},x_{b,r},\ell \cdot (2i-2) +2(b-1))$
induce a strict temporal path
from $x_{a,q}$ to $x_{b,r}$
(hence they do not violate the strict time constraint), since $2(b-1) \geq 2(a+1-1)= 2a > 2a-1$.
Moreover, note that if $(v_b, v_c)$ is an arc of
$P$, then
the corresponding temporal edges of $P'$
have timestamps greater than 
$\ell \cdot (2i-2) +2(b-1)$ 
(the first of these temporal edges in $P'$
is 
$(x_{b,r}, y_{b,c},\ell \cdot (2i-2) +2b ))$.
This ensure that $P'$ is an $\ell$-temporal path in $G$.

\end{enumerate}

Starting from $P$, we have defined an $\ell$-temporal path
in $G[V \setminus V']$, leading to a contradiction.

We conclude that $A'$ is a solution of \DMC{} on instance $(D, R)$ that contains at most $k$ arcs, thus concluding the proof.
\end{proof}

\subsection*{Proof of Theorem \ref{teo:hard1}}

\teohardOne*
\begin{proof}
The result follows from Lemma \ref{lem:DMCSt}, 
which implies we have designed an approximation
preserving reduction from \DMC{} to
\StTempSep{$\ell$}.
\DMC{}
is not approximable within factor
$2^{\Omega(\log^{1-\varepsilon}|N|)}$, for any constant $\varepsilon > 0$, unless $NP \subseteq ZPP$ \cite{DBLP:journals/jacm/ChuzhoyK09}.
Since $G$ contains $|V|$ vertices, with
$|V| = O(|N||A|)$ and $|A| = O(|N^2|)$,
it follows that $\log^{1-\varepsilon}|V| = \Omega(\log^{1-\varepsilon}|N|)$,
thus \StTempSep{$\ell$} is hard to approximate within factor $2^{\Omega(\log^{1-\varepsilon}|V|)}$, for any constant $\varepsilon> 0$, unless $NP \subseteq ZPP$.
\end{proof}

\subsection*{Proof of Corollary \ref{cor:hard1}}

\corhard*
\begin{proof}
We can modify the construction of instance $(G,s,z, \ell)$ so that each temporal path from $s$ to $z$ in $G$ is 
strict. In this way, the result can be extended to
\TempSep{$\ell$}.

Consider a temporal edge 
$(x_{i,q}, y_{i,j},\ell \cdot (2h-2) + 2i-1)$
($(y_{i,j}, x_{j,r},\ell \cdot (2h-2) + 2(j-1))$, respectively),
it is subdivided in two temporal edges 
$(x_{i,q}, w_{i,j,q},t_{q,1})$,
$(w_{i,j,q},y_{i,j},t_{q,2})$
($(y_{i,j}, w'_{i,j,r},t_{r,1})$,
$(w'_{i,j,r}, x_{j,r},t_{r,2})$, respectively)
such that $t_{q,1} < t_{q,2} < t_{r,1} < t_{r,2}$.
Timestamps $t_{q,1}$, $t_{q,2}$, $t_{r,1}$, $t_{r,2}$
are defined so that there is no temporal path
consisting of a temporal edge
(1) from $x_{i,q}$ to $y_{i,j}$ and 
(2) from $y_{i,j}$ to $x_{i,r}$, with $q,r \in [|A|+1]$
and $q \neq r$.
This ensures that any $\ell$-temporal path in $G$ traverses first vertices $x_{i,q}$
and then $x_{j,r}$ with $i <  j $.

As in the previous construction, 
given a solution $N'$ of \DMC{} on instance $(D,R)$,
we can compute in polynomial time 
an $(s,z,\ell)$-temporal separator of $G$
consisting
of $|N'|$ vertices $y_{i,j}$, $i,j \in |[|N|]$.
Similarly, we can prove that, given
an $(s,z,\ell)$-temporal separator of $G$,
we can assume that $V'$ contains only vertices
$y_{i,j}$, $i,j \in |[|N|]$. 
A solution of \DMC{} on instance $(D,R)$,
consists of the $|V'|$ arcs corresponding to
the vertices $y_{i,j} \in V'$.
This is motivated by the property that
there exists a path in $D$ connecting two
terminals if and only if 
there exists an $\ell$-temporal path in $G$. 

In this way we design an approximation preserving reduction from \DMC{}
to \StTempSep{$\ell$}, thus concluding the proof.
\end{proof}

\subsection*{Proof of Claim \ref{claim:easy}}
\claimeasy*
\begin{proof}
Consider the strict temporal path $P = [sv_1t_1, u_2v_2t_2, \dots 
u_h z t_k]$, with $tt(P) = t_k - t_1 +1$
and $t_k - t_1 + 1 \leq \ell$, hence  $t_k - t_1 < \ell$.
Since $P$ is strict, it holds that $t_i < t_{i+1}$, with $i \in [h]$, and $P$ consists of  $k \leq \ell$ temporal edges 
and traverses $k-1 \leq \ell-1$ vertices different from $s$ and $z$.
\end{proof}

\subsection*{Proof of Theorem \ref{TeoStrictApprox}}

\TeoStrictApprox*
\begin{proof}
First note that the solution $V'$  returned by Algorithm \ref{alg:ApproxStrict} is an
$(s,z,\ell)$-strict temporal separator, since if there exists a strict $\ell$-temporal path $P$ between $s$ and $z$ in $G[V \setminus V']$, it must contain
a vertex not in $\{s,z\}$; in this case Algorithm~\ref{alg:ApproxStrict} finds such a path $P$ and removes all the vertices
in $V(P) \setminus \{ s, z\}$.

Now, consider the vertices of $V'$. By construction,
they belong to a set $\mathcal{P}' = \{ P_1, \dots , P_h\}$ of strict temporal paths between $s$ and $z$
that have travel time at most $\ell$, where $P_j$
have been considered in the $j$-th iteration
of Algorithm~\ref{alg:ApproxStrict}, $j \in [h]$.
Furthermore, any two $P_i, P_j \in \mathcal{P}$, with $1 \leq i < j \leq h$,
are vertex disjoint (except for $s$ and $z$),
since all the vertices in $V(P_i)$ had been removed
when $V_j$ is considered by Algorithm~\ref{alg:ApproxStrict}.
Any $(s,z,\ell)$-strict temporal separator of $G$
must then remove at least $h$ vertices, one
for each strict temporal path in $\mathcal{P}$.

By Claim \ref{claim:easy}
each of the strict temporal path in $\mathcal{P}$
contains at most $\ell-1$ vertices
different from $s$ and $z$, thus 
$|V'| \leq h (\ell-1).$
Since  the strict temporal paths in $\mathcal{P}'$ are vertex disjoint and an optimal solution
$V^*$ must contain at $h$ vertices,
it follows that 
\[
|V'| \leq (\ell-1) |V^*|,
\]
thus concluding the proof.
\end{proof}

\subsection*{Proof of Corollary \ref{cor:hardapprox}}
\hardapprox*
\begin{proof}
We give an approximation preserving reduction
from \UniformHypergraph{} to \StTempSep{$\ell$}
similar to the one given in~\cite{DBLP:conf/isaac/HarutyunyanKP23} from \textsc{Set Cover}.
We recall that \UniformHypergraph{}, given
a $k$ uniform hypergraph $H = (U,S)$, where $U$  is a set of vertices
and $S$ is a collection of hyperedges (subset of $k$ vertices),
the goal is to find a subset $U' \subseteq U$ of minimum cardinality that intersects every hyperedge in $S$.

Assume $S$ contains hyperedges $e_1, \dots, e_q$. 
Moreover, we assume that the vertices in $U$
are ordered (the specific order is not relevant).
First, 
we define $\ell = k + 1$ and $\tau = k \cdot q$.
Now, we define $G= ( V, E, \tau)$.
$V$ contains $s$, $z$ and,
for each $u_i \in U$, $i \in |U|$, a corresponding vertex $v_i \in V$.
For each $e_i \in S$, $i \in [q]$, we define a strict $\ell$-temporal path between $s$ and $z$
that traverses the vertices corresponding to elements in $e_i$ according to their ordering (that is if 
$u_i < u_j$ then $v_i$ is traversed before $v_j$),
where the temporal edges of this
strict $\ell$-temporal path 
are defined in time interval
$I_i = [2 i \ell+1, (2i+1) \ell]$.
Note that, since each $|e_i| = k$, $i \in [q]$,
the corresponding strict $\ell$-temporal path 
in $G$ contains at most $k+1 = \ell$ temporal edges,
each assigned a distinct timestamp and that two intervals $I_i$,
$I_j$, with $1 \leq i < j \leq |U|$, are separated by 
at least $\ell$ timestamps.

Now, each strict $\ell$-temporal path defined in time interval $I_i$, $i \in [q]$, 
corresponds to a hypheredge $e_i$. Hence, a vertex $v_j$, $j \in [|U|]$, in  an $(s,z,\ell)$-strict temporal separator of $G$ corresponds to a vertex $u_i \in U$.
We have thus designed an approximation reduction from \UniformHypergraph{} to \StTempSep{$\ell$}
with $k = \ell-1$.
The result follows from the fact that
\UniformHypergraph{}
is not approximable within factor $k$ assuming the Unique Game Conjecture~\cite{khot2008vertex}, .
\end{proof}

\subsection*{Proof of Theorem~\ref{thm:cuthard}}

\thmcuthard*

\begin{proof}
Given an instance of \textsc{Multiway Cut} with graph $G$, terminals $v_1, \ldots, v_k$, we construct the temporal graph $H$ as described above.  
We claim that we can remove at most $w$ edges from $G$ to separate $v_1, \ldots, v_k$ if and only if $H$ admits an  $(s,z,\ell)$-temporal cut with at most $w$ temporal edges, where again $\ell = k - 1$.  
We comment on the case $\ell = 2$ at the end of the proof.

Suppose that $F \subseteq E(G)$ is a solution for $G$, so that $|F| \leq w$ and there is no path between any two distinct $v_i, v_j$, $i,j \in [k]$ and $i \neq j$, 
in $G - F$. 
Consider $F'$ consisting of the corresponding edges in $H$, i.e., $F' = \{(u,v,\ell) : uv \in F\}$.  We claim that $F'$ is an $(s,z,\ell)$-temporal cut.
Suppose for contradiction that $H - F'$ contains an $\ell$-temporal path $P'$ from $s$ to $z$.
By construction, the first temporal edge of $P'$ must be $(s, v_i, i)$ for some $i \in [k-1]$.  By the time constraints, the temporal edge $(v_i, z, \ell + i)$ cannot be used in $P'$, and the same holds for temporal edges $(v_{i'}, z, \ell + i')$ with $i'$ between $i$ and $k - 1$.  Therefore, the last temporal edge of $P'$ is either $(v_j, z, \ell + j)$ for some $j < i$, or it is $(v_k, z, \ell)$.  Either way, this means that $P'$ starts with $s$ and $v_i$, then uses only temporal edges copied from $G - F$ to reach $v_j$ or $v_k$, then goes to $z$.  This is a contradiction, since $F$ cuts all paths between $v_i$ and $v_j, v_k$ and the edges of $H$ corresponding to the temporal edges of $P'$ induce a path between two terminal in $G$.  
Thus $F'$ has at most $w$ edges and is a solution 
of \TempCut{$\ell$} on instance $H$.

Conversely, suppose that there is $F' \subseteq E(H)$ such that $|F'| \leq w$ and such that $H - F'$ contains no $\ell$-temporal path between $s$ and $z$. 
In $G$, remove the set of edges $F = \{uv : (u, v, \ell) \in F'\}$, which is well-defined since the only deletable edges of $H$ are those copied from $G$. 
Note that if there is a path between $v_i$ and $v_j$ in $G - F$, with distinct $i, j \in [k]$, then in $H - F'$ there is also a temporal path between $v_i$ and $v_j$ using only temporal edges copied from $G$.  
We argue that this is not possible.  Indeed:
\begin{itemize}
    \item 
    if there is a path between $v_i$ and $v_j$ in $H - F'$, where $i < j \leq k - 1$, then in $H - F'$ there is the temporal path that starts from $s$, then goes to $v_j$ using the edge $(s, v_j, j)$, then goes to $v_i$, then $z$ using $(v_i, z, \ell + i)$.  In this temporal path, the first
    temporal edge has time $j$, the temporal path from $v_j$ to $v_i$ has temporal edges at time $\ell > j$, and the last temporal edge has time $\ell + i > \ell$.  Hence this is a temporal path.  Since $j \geq i + 1$, the time spanned by the path is $i + \ell - j \leq i + \ell - (i + 1) = \ell - 1$.  We have thus formed an $\ell$-temporal path in $H - F'$, a contradiction.

    \item 
    if there is a temporal path between $v_i$ and $v_k$ in $H - F'$, where $i < k$, then there is the temporal path that starts in $s$, goes to $v_i$
    and then to $v_k$, then $z$
    in $H - F'$.  This is a temporal path because it starts with an edge at time $i \leq \ell$ and all other edges have time $\ell$, and since $\ell - i \leq \ell - 1$, we have again formed an $\ell$-temporal path in $H - F'$.

\end{itemize}
Since no $\ell$-temporal path between $v_i$ and $v_j$
exists in $H - F'$ for any distinct $i, j \in [k]$, no  path between $v_i$ and $v_j$ exists in $G - F$ either, and thus $F$ is a solution for $G$ with $|F| \leq w$.  This completes the proof, since we have described
an approximation preserving reduction 
and \textsc{Multiway Cut} is 
APX-hard~\cite{dahlhaus1992complexity}.
Note that when $k = 3$ terminals are given the \textsc{Multiway Cut} problem is still APX-hard~\cite{dahlhaus1992complexity}. In this case the reduction puts $\ell = 2$, which means that \TempCut{$\ell$} is APX-hard even when $\ell = 2$.  
\end{proof}

\subsection*{Proof of Lemma \ref{lem:EdgeShare}}

\lemEdgeShare*
\begin{proof}
Aiming at a contradiction, assume that a temporal path $p$ of $G([t+i, t+i+\ell-1]) - E'_t$
and a temporal path $p'$ of $G([t-j, t-j+\ell-1]) - E'_t$ are not temporal edge disjoint.
Thus there exists a temporal edge $e=(u,v,t_e)$ that belongs to both $p$ and $p'$.
Let $p_1$ ($p'_1$, respectively) be the subpath of $p$ (of $p'$, respectively) between $s$ and $u$;
let $p_2$ ($p'_2$, respectively) be the subpath of $p$ (of $p'$, respectively) between $v$ and $z$.
Then define $p^* = p_1 + e + p'_2$ (the concatenation of $p_1$, $e$ and $p'_2$); $p^*$ is a temporal walk between $s$ and $z$, since the following properties hold:

\begin{itemize}
    \item All the temporal edges of $p_1$ must have associated timestamps
    not larger than $t_e$ and at least $t+i$.

    \item The temporal edges of $p'_2$ must have associated timestamps
    not smaller than $t_e$ and at most $t-j+\ell-1$.

\end{itemize}

It must hold that
$t-j+\ell-1 \geq t+i$, since $t-j+\ell-1 \geq t_e$
and $t+i \leq t_e$. Moreover, it holds that
\[
(t-j+\ell-1) - (t+i) \leq \ell - 1
\]
since $(t+i) \leq (t-j+\ell-1) \leq (t+i+\ell-1)  $ 
and $(t+i+\ell-1)  - (t+i)   = \ell-1$.

Thus $p^* = p_1 + e + p'_2$ is a temporal walk between $s$ and $z$ in $G([t+i, t+\ell+i-1]) - E'_t$, $G([t-j, t'-j+\ell-1]) - E'_t$, thus also in 
$(G[t,t+\ell-1])$. 
Each vertex of $p^*$ has at most two occurrences in $p^*$ (at most one in $p'_1 + e$, at most one in $e + p_2$, since $p'_1 + e$ and $e + p_2$ are 
temporal paths).

Now, given the temporal walk $p^*$, we can compute a temporal path
between $s$ and $z$ in $G([t,t+\ell-1])$. Assume that $x_1, \dots x_n$
are the repeated vertices of $p^*$ and let $x_1$ be the first repeated vertex of $p^*$. 
Let $(x_1, y_1, t_1)$ and $(x_1, y_2, t_2)$ be the temporal edges outgoing from $x_1$ in $p^*$. Then we can compute a temporal walk in $G([t,t+\ell-1])$ by removing the subwalk from
$x_1$ to $x_1$ and replacing it with temporal edge $(x_1, y_2, t_2)$. Since $t_1 \leq t_2$, we have obtained
a temporal walk that contains a single occurrence of
$x_1$. By iterating this procedure, we obtain an 
$\ell$-temporal path between $s$ and
$z$ in $G([t,t+\ell-1])$ that is not cut by $E'_t$,
thus $E'_t$ is not a temporal cut of $G([t,t+\ell-1])$. 
This leads to a contradiction, thus proving the lemma.
\end{proof}

\subsection*{Proof of Lemma \ref{lem:EdgeShare2}}

\lemEdgeShareTwo*

\begin{proof}
Consider the graph $G([t_{i,j}+\ell/2, t_{i,j}+\ell/2-1]) - \bigcup_{p=1, q \in [w_p]}^{i-1} E'_{p,q}$ and 
let $E'(t_{i,j})$ be the minimum cut of 
$G([t_{i,j}+\ell/2, t_{i,j}+\ell/2-1]) - \bigcup_{p=1, q \in [w_p]}^{i-1} E'_{p,q}$
computed by Algorithm~\ref{alg:ApproxEdge}.

For each $j \in w_i$, since $E'(t_{i,j})$
is a minimum cut of 
$G([t_{i,j}+\ell/2, t_{i,j}+\ell/2-1]) - \bigcup_{p=1, q \in [w_p]}^{i-1} E'_{p,q}$,
it follows by Theorem~\ref{teo:MengersTemporal}
that there exists a set $P(t_{i,j})$ 
of temporal edge disjoint temporal paths in 
\[
G([t_{i,j}+\ell/2, t_{i,j}+\ell/2-1]) - \bigcup_{p=1, q \in [w_p]}^{i-1} E'_{p,q}.\]
Each temporal path in $P(t_{i,j})$ 
must be cut by $Opt$, either by some
temporal edge in $Opt(t_{p,q})$, with $p \in [i-1]$ or by $Opt(t_{i,j})$.
Thus
\[
|E'(t_{i,j})| \leq 
|Opt(t_{i,j})| + \sum_{p = i, q \in [w_i]}^{i-1} |Opt(t_{p,q})|.
\]
By Lemma~\ref{lem:EdgeShare},
it holds that each temporal path $p \in P(t_{i,j})$ and each temporal path $p' \in P(t_{i,q})$, with $j,q \in [w_i]$ and $j \neq q$, are temporal edge disjoint. 
Recall that sets $E'(t_{i,j})$ and $E'(t_{i,q})$
are disjoint, and sets
$Opt(t_{i,j})$ and $Opt(t_{i,q})$ are disjoint.
Moreover, if a temporal edge of a set $Opt(t_{b,a})$,
with $b < i $ and $a \in [w_b]$, cuts a temporal path
in $P(t_{i,j})$, it does not cut a path $P(t_{i,q})$,
as these two sets of temporal paths are temporal edge disjoint.
It follows that
\[
\sum_{j \in [w_i]} |E'(t_{i,j})| \leq 
\sum_{j \in [w_i]} |Opt(t_{i,j})| + \sum_{b\in [h], b < i, q \in [w_b] } |Opt(t_{b,q})|
\]
thus concluding the proof.

\end{proof}

\subsection*{Proof of Theorem \ref{logTapprox}}

\logTapprox*

\begin{proof}
By Lemma \ref{lem:EdgeShare2} we have that, for each
$i \in [h]$, it holds that
\[
\sum_{j \in [w_i]} |E'(t_{i,j})| \leq 
\sum_{j \in [w_i]} |Opt(t_{i,j})| + \sum_{b \in [h]:b <i, q \in [w_b]} |Opt(t_{b,q})|.
\]
Thus, since $h = \log_2 \tau$, 
it holds that
\[
|E'| = 
\sum_{i=1}^{\log_2 \tau} 
\sum_{j \in [w_i]} |E'(t_{i,j})| \leq 
\sum_{i=1}^{\log_2 \tau}
\left(
\sum_{j \in [w_i]} |Opt(t_{i,j})| + \sum_{b \in [h]:b <i, q \in [w_b]} |Opt(t_{b,q})|
\right).
\]
Since each $Opt(t_{i,j})$ contributes at most
$O(\log_2 \tau)$ times to the right term,
we have that
\[
|E'| = 
\sum_{i=1}^{\log_2 \tau} 
\sum_{j \in [w_i]} |E'(t_{i,j})| \leq 
\sum_{i=1}^{\log_2 \tau}
\left(
\sum_{j \in [w_i]} |Opt(t_{i,j})| + \sum_{b \in [h]:b <i, q \in [w_b]} |Opt(t_{b,q})|
\right), 
\]
which is bounded by
$\log_2 \tau \sum_{i \in [h],j \in [w_i]} |Opt(t_{i,j})|
= \log_2 \tau |Opt|$,
%
thus concluding the proof.
\end{proof}

\subsection*{Proof of Lemma \ref{lem-approx1}}
\lemapprox*
\begin{proof}
(1)Consider interval $I=[t,t+\ell-1]$, by construction
there exists an interval $J=[a+1,a+2\ell]$ in $P_1$ that
contains $t$. Now, if $t \leq a + \ell$, then
$I$ is contained in $J$. If this is not the case,
$t \in [a+\ell+1, a +2\ell]$ and
$I$ is contained in $[a+\ell+1, a+ 3\ell]$,
that belongs to $P_2$.

(2) The result follows from the fact that
$I \cap I'= \emptyset$, hence each temporal edge 
defined in $G[(I])$ is not defined in $G[(I'])$ (and vice versa).
\end{proof}

\subsection*{Proof of Corollary \ref{corEdgeApprox}}

\corEdgeApprox*

\begin{proof}
Consider an optimal solution $OPT$ of \TempCut{$\ell$}
and let $OPT(G(I)))$ be the set of temporal edges
in $OPT$ defined in some timestamp of $I$.
 The approximation algorithm with input $G(I)$ returns an approximated solution $E(I)$.
By Theorem \ref{logTapprox}, since each interval consists
of $2 \ell$ timestamps, it holds that
\[
|E(I)| \leq \log_2(2 \ell) |OPT(G(I))|.
\]

Consider two different intervals $I$, $I'$ of a set $P_i$, $i \in \{1,2\}$. Since $I$ and $I'$ are disjoint, 
the temporal edges in sets $OPT(G(I))$ and $OPT(G(I'))$ are disjoint, thus
\begin{equation}  
\label{eq:approx-l}
\sum_{I \in P_i} |E(I)| \leq \log_2(2 \ell) \sum_{I \in P_i} |OPT(G(I))| \leq \log_2( 2 \ell) |OPT|
\end{equation}

Now, since Equation \ref{eq:approx-l} holds
for both $P_1$ and $P_2$, it follows that
\begin{equation}  
\label{eq:approx-l-2}
\sum_{i=1}^2\sum_{I \in P_i} |E(I)| \leq  \log_2(2 \ell) 
\sum_{I \in P_1}
|OPT(G(I))| +  \log_2(2 \ell) \sum_{I \in P_2}|OPT(G(I))| \leq 2 \log_2(2 \ell) |OPT|.
\end{equation}
Now, the set $E'$ returned by the approximation algorithm is a feasible solution of the problem
(that is $E'$ cuts every temporal path of time travel at most $\ell$). Indeed by Lemma \ref{lem-approx1} each interval $I$ of length $\ell$ is contained in an interval of $P1$ or $P_2$,
and Algorithm \ref{alg:ApproxEdge} computes a 
cut in every $G(I)$.

\end{proof}

\end{document}